\documentclass[12pt]{article}
\usepackage{amsfonts,amssymb,amsthm,eucal,amsmath,verbatim,color,bbm}
\usepackage{pdfpages}
\usepackage{graphicx}
\graphicspath{ {images/} }

\allowdisplaybreaks
\newtheorem{thm}{Theorem}
\newtheorem{cor}[thm]{Corollary}
\newtheorem{lemma}[thm]{Lemma}
\newtheorem{prop}[thm]{Proposition}

\newcommand{\R}{\mathbb{R}}
\newcommand{\Z}{\mathbb{Z}}

\newcommand{\E}{\mathbb{E}}

\newcommand{\inprod}[2]{\left\langle #1, #2 \right\rangle}

\renewcommand{\P}{\mathbb{P}}

\newcommand{\ee}{\mathbb{E}}

\newcommand{\s}{\mathbb{S}}
\newcommand{\tr}{\mathrm{Tr\,}}

\newcommand{\1}{\mathbbm{1}}

\allowdisplaybreaks

\begin{document}
\title
{Asymptotics of mean-field $O(N)$ models}
\author{Kay Kirkpatrick and Tayyab Nawaz\thanks{Both authors partially supported by National Science Foundation (NSF) CAREER award DMS-1254791,
and NSF grant 0932078 000 while in residence at Mathematical Sciences Research Institute (MSRI) during
Fall 2015.} \\ \\
Department of Mathematics, University of Illinois at Urbana-Champaign\\
1409 W. Green Street, Urbana, IL 61801, USA}


\maketitle

\begin{abstract}
We study mean-field classical $N$-vector models, for integers $N\ge 2$. We use the theory of large deviations and Stein's method to study the total spin and its typical behavior, specifically obtaining non-normal limit theorems at the critical temperatures and central limit theorems away from criticality. Important special cases of these models are the XY ($N=2$) model of superconductors, the Heisenberg ($N=3$) model (previously studied in \cite{KM} but with a correction to the critical distribution here), and the Toy ($N=4$) model of the Higgs sector in particle physics.
\\ \\
{\bf keywords:} Mean-field, Rate function, Total spin, Limit theorem, Phase transition.

\end{abstract}

\section{Introduction}

In statistical mechanics, mean-field models are often the starting point for understanding the behavior of the underlying physical systems, at least in high dimensions. In particular, we can use large deviations to study the asymptotics of physical quantities such as magnetization (total spin in our terminology) in models such as XY for superconductors, Heisenberg for ferromagnets \cite{KM}, or the Toy model for particle physics. There is a family of models generalizing these important special cases, namely the  mean-field classical $N$-vector spin models, where each spin $\sigma_i$ is in an $N$-dimensional unit hypersphere, at a lattice site or in our case at a complete graph vertex $i$ among $n$ vertices. Then in the absence of an external field, each microstate or spin configuation $\sigma = {(\sigma_1,\sigma_2,...,\sigma_n)}$ in the configuration space $\Omega_n = (\s^{N-1})^n$ has a Hamiltonian defined by:

$$H_n(\sigma) = - \sum_{i,j} J_{i,j}\inprod{\sigma_i}{\sigma_j}.$$
For the mean-field models defined on the complete graph, every two vertices $(i,j)$ are adjacent and the interaction between $\sigma_i$ and $\sigma_j$ is given by the constant $J_{i,j} = \frac{1}{2n}$, which can be viewed as an averaged interaction.

The $N=1$ case of the mean-field $N$-vector model is the Curie-Weiss model, which approximates the Ising model well for higher dimensions. The normalized total spin in the
Curie-Weiss model has a non-Gaussian law in the non-critical
regimes, and a law that converges to the distribution with density
proportional to $e^{-x^4/12}$ at the critical temperature (Ellis and Newman \cite{EN}).  Chatterjee and Shao  proved that the total spin at the critical temperature for this non-central limit theorem satisfies an error bound of order $1/\sqrt{n}$ \cite{CS}.
The XY model on a two-dimensional lattice is especially interesting for applications to superconductors \cite{PWA}, but challenging to study its phase transition rigorously. 

For instance, the Mermin-Wagner theorem states that in two spatial dimensions, such a continuous symmetry cannot be broken spontaneously at any finite temperature \cite{MW}. This implies that the XY model on a two-dimensional lattice cannot have an ordered phase at low temperature like the Ising model does. Stanley and Moore provided some evidence that this system has a phase transition but it can't be of usual type with finite mean magnetization below the critical temperature \cite{SHE,MMA}.  But still in two dimensions it has an infinite-order transition names after Kosterlitz and Thouless (KT), who proved that there is phase transition from bound vortex-antivortex pairs at low temperature to unpaired vortices and anti-vortices at some critical temperature \cite{KT}. Above the transition temperature $T_{KT}$ correlations between spins decay exponentially as usual, with some correlation length. They also showed that this system does not have any long-range order as the ground state is unstable against low-energy spin-wave excitations.  They further proved that this is a low-temperature quasi-ordered phase with a correlation function that decreases with the distance like a power, which depends on the temperature. 

Because the two-dimensional lattice XY model is challenging, the mean-field case is often the starting point for rigorous analysis of these spin models, a case that can be thought of as a large-dimensional limit of nearest-neighbor lattice models, or as an infinite limit for complete graph models.
It is known that the classical $N$-vector model with spins in $\s^{N-1}$, in the large-dimensional ($d \to \infty$) limit on the lattice $\Z^d$, has the critical inverse temperature $\beta_c=N$ \cite{KS}. This limiting case is thought to approximate high-dimensional models well because magnetization goes to zero below the critical temperature for all $d$, and the magnetization goes to the correct limit above the critical temperature as $d \to \infty$.
 
In this paper we study mean-field $N$-vector models for positive integers $N$, in the infinite limit for complete graphs. One reason for studying different dimensions collectively is because we can generate a family of solutions expressed in terms of modified Bessel functions, which can then be used for asymptotic analysis of the mean-field models. We have the following results: 

\begin{enumerate}
\item Section 2 contains Large deviation principles (LDPs) for the total spin (with rate functions) and the empirical spin distribution (with relative entropies) for each non-negative inverse temperature. 

\item Section 3 states the limit theorems for total spin in each phase, which are proved in sections 4 through 6.

\item Section 4 proves a central limit theorem in the subcritical (disordered) phase for the total spin, with Stein's method. 

\item Section 5 proves a CLT in the supercritical (ordered) phase for the total spin with a different scaling, again with Stein's method. 

\item Section 6 we derive a non-normal limit theorem for the total spin at the critical temperature, with limiting density of the squared length proportional to $t^{\frac{N-2}{2}} e^{-\frac{1}{4 N^2 (N+2)} t^2}$.

\item The Appendix contains some technical details including calculus for the free energy, and abstract results for the non-normal Stein's method application.

\end{enumerate}

\section{The Setup and Large Deviations}

We consider, for $N$ which is a fixed positive integer, the mean-field classical $N$-vector, or $O(N)$, model on a complete graph $K_n$ with $n$ vertices (these are isotropic models, meaning no external magnetic field). At each site $i$ on the graph is a spin $\sigma_i$ in $\Omega = \s^{N-1}$, so the state (or configuration) space is $\Omega_n = (\s^{N-1})^n$ with product measure $P_n$ from the uniform probability measure on $\s^{N-1}$. For these models, and then the mean-field Hamiltonian energy $H_n : \Omega_n \to \R$ is defined by: 
\[ H_n (\sigma) :=  -\frac{1}{2n} \sum_{i,j=1}^n\inprod{\sigma_i}{\sigma_j}. \]
The energy per particle is $h_n(\sigma)=\frac{1}{n}H_n(
\sigma)$, and the canonical ensemble, or Gibbs measure, is the probability measure 
$P_{n,\beta}$ on $\Omega_n$ with density (with
respect to $P_n$): 
$$f(\sigma):=\frac{1}{Z}e^{-\beta H_n(\sigma)} =
\frac{1}{Z}\exp\left(\frac{\beta}{2n}  \sum_{i,j=1}^n\inprod{\sigma_i}{\sigma_j}
\right).$$
The (normalizing) partition function is given by:\\
\[Z = Z_n (\beta) = \int_{\Omega_n} \exp\left( \frac{\beta}{2n} \sum_{i,j=1}^n\inprod{\sigma_i}{\sigma_j} \right) dP_n.\]
We will call $M_1(\s^{N-1})$ the space of probability measures on $\s^{N-1}$ with the weak-* topology. 

Now we are interested in studying the behavior of the important physical quantity of total spin $S_n := \sum_{i=1}^n \sigma_i$ in terms of the inverse temperature $\beta$, distributed according to the Gibbs measures. First we present a proposition stating the large deviation principle for the empirical spin distribution for non-interacting particles (disordered infinite-temperature case) $\beta = 0$, a proposition that is a special case of Sanov's theorem.

\begin{prop}\label{prop1}
 If $P_n$ is the $n$-fold product of uniform
  measure on 
$\s^{N-1}$, $\mu_{n,\sigma}=\frac{1}{n}\sum_{i=1}^n\delta_{\sigma_i}$ is the empirical spin distribution, and $\Gamma$ a Borel subset of $M_1(\s^{N-1})$, then
\begin{align*}
-\inf_{\nu\in\Gamma^\circ}H(\nu\mid\mu)  &\le  \liminf_{n\to\infty}\frac{1}{n}
\log P_n[\mu_{n,\sigma}\in\Gamma]  \\
&\le  \limsup_{n\to\infty}\frac{1}{n}\log P_n[\mu_{n,\sigma}\in\Gamma]\le-\inf_{
\nu\in\overline{\Gamma}}H(\nu\mid\mu);
\end{align*}
i.e., the random measures $\mu_{n,\sigma}$ satisfy an LDP with rate function 
$H(\cdot\mid\mu)$, the relative entropy, defined for a probability measure $\nu$ on $\s^{N-1}$, with respect to uniform measure $\mu$, by:
\[H(\nu\mid\mu):=\begin{cases}\int_{\s^{N-1}}f\log(f)d\mu& \text{if} \,\, f:=\frac{d\nu}{d\mu} 
\,\,\text{exists};\\\infty&\text{otherwise}.\end{cases}\]
\end{prop} 
In particular, at infinite temperature $\beta = 0$, the unique minimizer of the rate function is the uniform measure $\mu$, meaning that spins are uniformly and independently distributed on the sphere, with no preferred direction in this disordered phase.

Now we state the general case  $\beta \ge 0$, which follows from abstract results of Ellis, Haven, and Turkington (\cite{EHT}, Theorems 2.4 and 2.5):

\begin{thm}\label{mainLDP} If $\beta \ge 0$ then the empirical spin distributions $\mu_{n,\sigma}$
satisfy an LDP on $M_1(\s^{N-1})$ with rate function: 
\begin{equation}\label{beta_rf}
I_\beta(\nu):=H(\nu\mid\mu)-
\frac{\beta}{2}\left|\int_{\s^{N-1}}xd\nu(x)\right|^2-\varphi(\beta),
\end{equation}
where $\varphi$ is the free energy defined by $\varphi(\beta):=-\lim_{n\to\infty}
\frac{1}{n}\log Z_n(\beta)$, which exists and is given by the alternative formula:
\begin{equation}\label{free-energy}\varphi(\beta)=
\inf_{\nu\in M_1(\s^{N-1})}\left[ H(\nu\mid\mu)-
\frac{\beta}{2}\left|\int_{\s^{N-1}}xd\nu(x)\right|^2\right].\end{equation}

\end{thm}

\noindent {\bf Remarks on notation:} 
\begin{enumerate}

\item Throughout this paper we will write the rate function as $I_\beta$ where the subscript is the real, non-negative inverse temperature $\beta = 1/(k_B T)$, where $T$ is temperature and $k_B$ is the Boltzmann constant; we will write the modified Bessel function of the first kind as $I_n$ where $n$ is an integer. 

\item The average magnetization $m = \ee[\frac1n S_n] = \ee[\frac1n \sum_i \sigma_i]$ for the $N$-vector spin model can be calculated by differentiating the partition function and comes out to be, in terms of Bessel functions:
\[|m| = \frac{I_\frac{N}{2}(|x|)}{I_{\frac{N}{2}-1}(|x|)}.\]

\item We can also verify that the rate functions $I_\beta$ are always nonnegative, as illustrated in Figure \ref{fig:1} below for $ N = 2 $ in three representative cases of $\beta$. \\ \\

\begin{figure}
\centering
\includegraphics[width=4in]{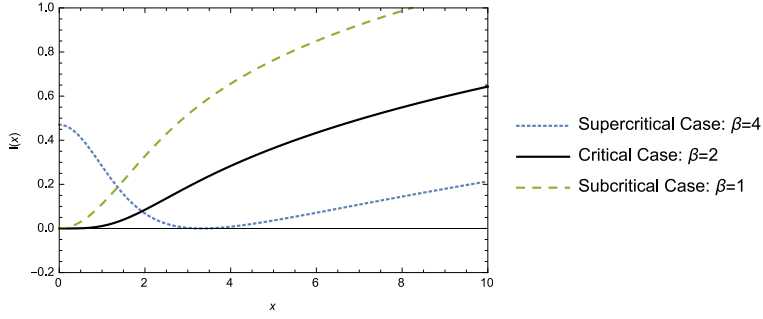}
\caption{Rate Function $I_\beta$ for the mean-field XY Model}
\label{fig:1}
\end{figure}
\end{enumerate}

\begin{thm}\label{freeenergy}
For dimension $ N $ , the free energy $\varphi$ has the formula:
\[ \varphi(\beta) = \begin{cases} 0, \quad \quad \quad \quad \quad \quad \quad \quad \text{ if } \beta <N , \\ \Phi_\beta(\Upsilon^{-1}(\beta)),  \quad \quad \quad \;\, \text{ if } \beta\ge N, \end{cases}\] 
where for $r = \Upsilon^{-1}(\beta)$,
\[\Phi_\beta(r)= r \frac {I_\frac{N}{2}(r)}{I_{\frac{N}{2}-1}(r)}+\log\left[\frac{A_{N}}{A_{N-1}} \frac {r^{\frac{N}{2}-1}}{ B_N \pi I_{\frac{N}{2}-1}(r) }\right]-\frac{\beta}{2}\left(\frac {I_\frac{N}{2}(r)}{I_{\frac{N}{2}-1}(r)}\right)^2 ,\]
with  \[A_N := \frac{2  \pi^\frac{N}{2}}{\Gamma{\left(\frac{N}{2}\right)}}, \quad \quad \Upsilon(r) = \Upsilon_N(r) := r \frac{I_{\frac{N}{2}-1}(r)}{I_\frac{N}{2}(r)} = \beta,\] and 
\[B_N= \begin{cases} \prod_{k=0}^{\frac{N}{2}-1} |2k-1| , \quad \quad \quad \quad \quad \quad \quad \quad \text{ if N even}, \\ \\ 2^{\frac{N-3}{2}} (1)_{\frac{N-3}{2}},  \quad \quad \quad \quad \quad \quad \quad \quad \quad\;\, \text{ if N odd}, \end{cases}\]
where $(a)_n$ is pochhammer symbol.
 In particular, we find the critical threshold $\beta =N$, and we can check by calculating limits that $\varphi$ and $\varphi'$ are continuous, implying that the phase transition is continuous.
 
\end{thm}
We can precisely describe the phase transition in free energy as follows:
 \begin{itemize}
 \item For $0 \le \beta \leq \beta_c = N$, we obtain unique global minima for free energy at the origin with a zero magnetization.
 \item For $\beta \geq \beta_c = N$, we have infinitely many global minima for the free energy which can be approximated graphically. Furthermore, the minima in this case are identical with non-zero magnetization. 

 \end{itemize}
We can deduce from Proposition \ref{prop1} the following Cram\'er-type LDP for the average spin $M_n := \frac1n S_n =\frac1n \sum_{i=1}^n \sigma_i$ in the noninteracting case $\beta=0$ (or alternatively prove this Cram\'er theorem directly for random vectors on the hypersphere). 

\begin{cor}\label{spinLDP} If $\{\sigma_i\}_{i=1}^n$ are i.i.d.\ uniform random points on $\s^{N-1}\subseteq\R^N$, then for $r = |x|$, the average spins $M_n$  satisfy an LDP with rate function $I$:

\[P_n \left( M_n\simeq x \right) \simeq e^{-nI(r)},\]
where $I(r) = \Phi_0(r)$ from Theorem \ref{freeenergy}.
\end{cor}
Similarly we have a Cram\'er-type result for the interacting case $\beta > 0$ as follows:
\begin{prop} \label{SpinLDPM} If $M_n=M_n(\sigma):=\frac1n \sum_{i=1}^n \sigma_i$ and  $P_{n,\beta}$ is the Gibbs measure as defined above, then for a Borel set $\Gamma\in\R$,
\begin{align*}
-\inf_{x\in\Gamma^\circ}I_\beta(x) &\le \liminf_{n\to\infty}\frac{1}{n}
\log
P_{n,\beta}\left[\beta M_n\in\Gamma\right] \\ &\le \limsup_{n\to\infty}\frac{1}{n}\log
P_{n,\beta}\left[\beta M_n\in\Gamma\right] \le -\inf_{
x\in\overline{\Gamma}}I_\beta(x)
\end{align*}

where for $r = |x|$,

\[I_\beta(r)= \Phi_\beta(r)-\inf_{\nu\in M_1(\s^{N-1})} \Phi_\beta(r) ,\]

and $\Phi_\beta(r)$ is defined in Theorem \ref{freeenergy}.

\end{prop}
Using similar reasoning as \cite{KM}, we consider the measure $\nu_g$ with the density function \[f(x_1,x_2,...,x_N)=g(x_N)\] which is increasing in $x_N$. This gives:
\begin{equation*}\begin{split}
H(\nu_g\mid\mu)&=\int_{\s_{N-1}} f(x_1,x_2,...,x_N)\log[
f(x_1,x_2,...,x_N)] dx_1 dx_2...dx_N
\\&=\frac{1}{A_{N}}\int_{0}^{\pi}\int_{0}^{ \pi}...\int_{0}^{\pi}\int_{0}^{2 \pi} g(\cos(\theta_{N-1}))\log[
g(\cos(\theta_{N-1}))]\\&~~~~~~~~~~~~~~~~~~~~~~~~~~~~~~~~~~ \prod_{k=2}^{N-1} Sin^{k-1} (\theta_k) d\theta_1 d\theta_2...d\theta_{N-1}\\&=\frac{A_{N-1}}{A_{N}}\int_{0}^{\pi} g(\cos(\theta_{N-1}))\log[
g(\cos(\theta_{N-1}))] Sin^{N-2} (\theta_{N-1})  d\theta_{N-1}\\&=\frac{A_{N-1}}{A_{N}}\int_{-1}^1g(x_N)\log[
g(x_N)] \left(1-x_N^2\right)^\frac{N-3}{2} dx_N.
\end{split}\end{equation*}
Similarly, if $e_N$ is the unit coordinate vector in direction $x_N$, then
\[\int_{\s^{N-1}} vd\nu_g(v)= e_N \frac {A_{N-1}}{A_{N}}\int_{-1}^1
x_N g(x_N)\left(1-x_N^2\right)^\frac{N-3}{2}dx_N.\]
Now we are interested in minimizing the functional
\[\frac{A_{N-1}}{A_{N}}\int_{-1}^1g(x_N)\log[
g(x_N)] \left(1-x_N^2\right)^\frac{N-3}{2} dx_N-\frac{\beta}{2} \left(\frac {A_{N-1}}{A_{N}}\int_{-1}^1
x_N g(x_N)\left(1-x_N^2\right)^\frac{N-3}{2}dx_N\right)^2\]
under the following constraint: $g:[-1,1]\to\R_+$ increasing and such that \[\frac{A_{N-1}}{A_{N}}\int_{-1}^1g(x_N) \left(1-x_N^2\right)^\frac{N-3}{2} dx_N=1.\] 
We can also write the functional under consideration in terms of entropy as follows:

\begin{equation*}\begin{split}
& \frac{A_{N-1}}{A_{N}}\int_{-1}^1g(x_N)\log[
g(x_N)] \left(1-x_N^2\right)^\frac{N-3}{2} dx_N \\&=\frac{A_{N-1}}{A_{N}}\int_{-1}^1  g(x_N)\log\left[\frac{
g(x_N)A_{N}}{A_{N-1}}\right] \left(1-x_N^2\right)^\frac{N-3}{2} dx_N+ \log\left[\frac{A_{N-1}}{A_{N}}\right] \\& = -\xi \left (\frac{g A_{N}}{A_{N-1}}\right)+\log\left[\frac{A_{N-1}}{A_{N}}\right],
\end{split}\end{equation*}
where $\xi \left(\phi\right)$ is the entropy of the density $\phi$.

Now using constrained entropy optimization (Theorem 12.1.1 in
\cite{CT}), we fix $c \in [0,1]$ and minimize the above quantity
over the measures $\nu\in M_1(\s^{N-1})$ such that $\left|\int_{\s^{N-1}} xd\nu(x)\right| = c.$ \\

\begin{prop}\label{entropy_max}
Consider the set of functions $h:[-1,1]\to\R_+$ such that
\begin{enumerate}
\item $\int_{-1}^1h(x_N)\left(1-x_N^2\right)^\frac{N-3}{2} dx_N=1$, and 
\item $\left|\int_{-1}^1 x_N h(x_N) \left(1-x_N^2\right)^\frac{N-3}{2} dx_N\right|=c$.
\end{enumerate}
Then in the set of functions satisfying these conditions, $h^*(x)=a e^{b x}$ uniquely minimizes the quantity \[\frac{A_{N-1}}{A_{N}}\int_{-1}^1g(x_N)\log[
g(x_N)] \left(1-x_N^2\right)^\frac{N-3}{2} dx_N.\]
\end{prop}
Now we will use the conditions $(a)$ and $(b)$ to find the values of the parameters $a$ and $b$ for the function $h^*$. The first condition leads us to the following two subcases: for $ N  $ even, 
\[1=\int_{-1}^1 a e^{b  x_N}\left(1-x_N^2\right)^\frac{N-3}{2} dx_N=\frac {\left(\prod_{k=0}^{\frac{N}{2}-1} |2k-1| \right) a \pi I_{\frac{N}{2}-1}(b) }{b^{\frac{N}{2}-1}},\]
which implies 
\[ a = \frac{b^{\frac{N}{2}-1}} {\left(\prod_{k=0}^{\frac{N}{2}-1} |2k-1| \right) \pi I_{\frac{N}{2}-1}(b) }\]
Now using second condition, For $N $ even, we have
\begin{equation*}\begin{split}
c&=\int_{-1}^1 x_N h(x_N) \left(1-x_N^2\right)^\frac{N-3}{2} dx_N=\frac {\left(\prod_{k=0}^{\frac{N}{2}-1} |2k-1| \right) a \pi I_{\frac{N}{2}}(b) }{b^{\frac{N}{2}-1}} = \frac{I_{\frac{N}{2}} (b)}{I_{\frac{N}{2}-1}(b) }.
\end{split}\end{equation*}
Let $g^*=\left(\frac{A_{N-1}}{A_{N}}\right) h^*$;  and $c\in[0,1]$ with $g^*$ increasing corresponds
to considering all $b\in[0,\infty)$.
Now we have to minimize for $ N $ even,
\begin{equation}\begin{split}\label{tbm2}
& \frac{A_{N-1}}{A_{N}}\int_{-1}^1g(x_N)\log[
g(x_N)] \left(1-x_N^2\right)^\frac{N-3}{2} dx_N-\frac{\beta}{2} \left(\frac {A_{N-1}}{A_{N}}\int_{-1}^1
x_N g(x_N)\left(1-x_N^2\right)^\frac{N-3}{2}dx_N\right)^2\\&=
b \frac {I_\frac{N}{2}(b)}{I_{\frac{N}{2}-1}(b)}+\log\left[\frac{A_{N}}{A_{N-1}} \frac {b^{\frac{N}{2}-1}}{\left(\prod_{k=0}^{\frac{N}{2}-1} |2k-1| \right)\pi I_{\frac{N}{2}-1}(b)}\right]-\frac{\beta}{2}\left(\frac {I_\frac{N}{2}(b)}{I_{\frac{N}{2}-1}(b)}\right)^2 \\&=: \Phi_\beta(b)
\end{split}\end{equation}
over all $b\in[0,\infty)$. Using same approach we can derive the expression for $\Phi_\beta(b)$, for $N$ odd and this comes out to be
\[ \Phi_\beta(b) = b \frac {I_\frac{N}{2}(b)}{I_{\frac{N}{2}-1}(b)}+\log\left[\frac{A_{N}}{A_{N-1}} \frac {b^{\frac{N}{2}-1}}{\left(2^{\frac{N-3}{2}} (1)_{\frac{N-3}{2}} \right)\pi I_{\frac{N}{2}-1}(b)}\right]-\frac{\beta}{2}\left(\frac {I_\frac{N}{2}(b)}{I_{\frac{N}{2}-1}(b)}\right)^2,\] where $(a)_n$ represents pochhammer symbol. This is now one-dimensional problem which is studied in Lemma \ref{cal-phi} in Appendix. We deduce that $\beta_c = N$ is the critical inverse temperature. Also for $\beta < N $ we have the uniform distribution as the only macrocanonical state whereas for $\beta > N$ we have a family (parametrized by the circle) of distributions with a preferred direction (and converging to a family of points masses, each concentrated on a perfectly preferred direction as $\beta \to \infty $). We can state the following theorem using our calculations from Lemma \ref{cal-phi}:

\begin{thm}\label{freeenergyresults}

\begin{enumerate}
\item For $\beta\le N$, the expression \eqref{tbm2}  is minimized for $ b = 0$, then the corresponding $ a = 1$, so that the minimizing function $h^* = 1$ and hence the canonical macrostates in the subcritical case are uniform:  $\mathcal{E}_\beta=\{\mu\}$.
\item In the supercritical case, $\beta>N$, the minimizing $b$ for the expression \eqref{tbm2} is
  the unique strictly positive solution to 
\[b=\beta \left(\frac {I_\frac{N}{2}(b)}{I_{\frac{N}{2}-1}(b)}\right),\]which moreover
has limit $\lim_{\beta \downarrow \beta_c} b=0$.The macrostates $\mathcal{E}_\beta$ are given by
$\{\nu_x\}_{x\in\s^{N-1}},$ where $\nu_x$ is the probability measure with density which is
symmetric about the pole at $x$, with density $g_x:[-1,1]\to\R$ in the
$x$-direction given by $\left(\frac{A_{N-1}}{A_{N}}\right) a e^{b x}$ with $b$ as above.
\end{enumerate}
\end{thm}

\section{Limit Theorems for the Total Spin}
We study the total spin in the subcritical, critical, and supercritical regimes, proving central and non-central limit theorems for the total spin, holding $N$, the dimension of the spin, fixed.  In this section we state these limit theorems for each regime and give the proofs of each one of these in next few sections.

In the subcritical regime $0 \le \beta < N$, the spins are weakly correlated and hence can be treated similar to the independent case $\beta=0$. The average magnetization of the system is very small and goes to zero with increasing number of spins $n \to \infty$ for this high temperature regime. In this regime we have the following multivariate central limit theorem, and in particular, the macrostate is the uniform measure on the hypersphere.
\begin{thm} \label{subcrit_limit} In the subcritical regime $\beta<N$, the random variable $W_n$ is defined as follows:
$W_n=\sqrt{\frac{N-\beta}{n}}
\sum_{i=1}^n\sigma_i$
Then 
\[\sup_{g:L(g),M(g)\le 1}|\E g(W_n)-\E g(Z)|\le\frac{c_\beta}{\sqrt{n}}\]
where $c_\beta$ is a constant depending only on $\beta$, 
$L(g)$ is the Lipschitz constant of $g$, $M(g)$ is the
maximum operator norm of the Hessian of $g$, and $Z$ is a standard
Gaussian random vector in $\R^N$.
\end{thm}
Remark: Our rate of convergence for Theorem \ref{subcrit_limit} is sharper than \cite{KM}, which had a factor of $\log(n)$ in the numerator, based on an argument of Leslie Ross \cite{LAR}. The supremum in Theorem \ref{subcrit_limit} is a metric for the topology of weak-* convergence and convergence in mean on the space $M_1(\s^{N-1})$ of probability measures on the hypersphere.

In the supercritical regime $\beta >N$, the spins align to some extent: For smaller values of  $\beta > N$, the spins show a slight preference for a particular (random) direction, whereas for large $\beta$, the spins align strongly. Consider a small interval $\Gamma$ containing $b$, now using the fact that
 $\inf_{x\in\overline{\Gamma}}I_\beta(x) = b $ and Proposition \ref{SpinLDPM}, we conclude that $|S_n|$ has a high probability of being close to $b n/\beta$. Here all points on the hypersphere of radius $b n/\beta$ will have equal probability due to symmetry. Using an argument similar to \cite{KM}, we consider the fluctuations of squared-length of total spin, i.e., we consider the following random variable:
 \begin{equation}\label{W-def}
W_n:=\sqrt{n}\left[\frac{\beta^2}{n^2b^2}\left|\sum_{j=1}^n\sigma_j 
\right|^2-1\right].\end{equation}
 In Section \ref{S:supcrit}, we prove that this $W_n$ satisfies the following central limit theorem.

\smallskip

\begin{thm}\label{T:supcrit_CLT}
If $W_n$ is as defined in \eqref{W-def} and $b$ is the solution of $b-\beta f(b) = 0,$ 
where \[f(b)=\frac {I_\frac{N}{2}(b)}{I_{\frac{N}{2}-1}(b)},\] 
then there is a constant $c_\beta$ depending on
$\beta>N$ only, such that if $Z$ is a centered normal random variable with variance 
\[Var(\sigma)=\frac{4\beta^2}{\left(1-\beta f'(b)\right)b^2}
\left[1-\frac{N-1}{N}\frac{\left(I_{\frac{N}{2}-1}(b)+I_{\frac{N}{2}+1}(b)
\right)}{I_{\frac{N}{2}-1}(b)}-\left(\frac{I_{\frac{N}{2}}(b)}{I_{\frac{N}{2}-1}(b)}\right)^2\right]\]
 then
\[d_{BL}(W_n,Z)\le c_\beta\left(\frac{\log(n)}{n}\right)^{1/4}.\]
Here the bounded Lipschitz distance $d_{BL}(X,Y)$ between random
variables $X$ and $Y$ is:  
\begin{equation}\label{bldist}d_{BL}(X,Y):=\sup\left\{\Big|\E h(X)-\E
    h(Y)\Big|:\|h\|_\infty\le 1, M_1(h)\le 1\right\},\end{equation}
where $\| \cdot \|_\infty$ is the supremum norm  and $L(\cdot)$ is the Lipschitz constant as before. 
\end{thm}

We can obtain the complete asymptotic behavior of the total spin without using conditioning (as in e.g., \cite{EN}) by using instead the rotational invariance of the total spin, a strategy adapted from \cite{KM}.

In Section \ref{S:critical}, we prove the following nonnormal limit theorem for the random variable defined by 
\begin{equation}\label{Wsl-def}
W_n:=\frac { c_N |S_n|^2}{n^\frac{3}{2}}.\end{equation}
at the critical temperature $\beta=N$.  Because of symmetry of the total spin this leads us to the limiting picture in the critical case. The critical limiting density function $p$ (defined below) is obtained using Stein's method similar to \cite {CS,KM}. 

\begin{thm}\label{T:limit_crit}
If we consider the critical temperature $\beta=N$, and $W_n$ as defined by \eqref{Wsl-def} , and if  $X$ is the random variable with the density
 \[p(t)=\begin{cases}\frac{1}{z} t^{\frac{N-2}{2}} e^{-\frac{1}{4 N^2 (N+2)} t^2}&t\ge
 0;\\0&t<0,\end{cases} ,\]
where  $z$ is normalizing constant and  $c_N$ is such that $\E W_n=1$, then there exists a universal constant $C$ such that 
\[\sup_{\substack{\|h\|_\infty\le 1, \,\|h'\|_\infty\le 1\\\|h''\|_\infty\le
1}}\big|\E h(W_n)-\E h(X)\big|\le\frac{C\log(n)}{\sqrt{n}}.\]
\end{thm}

\section{The Subcritical Phase}\label{S:subcrit}

This section has the proof of Theorem \ref{subcrit_limit}, the limit theorem for $S_n$ in the
disordered phase. We start by calculating the variance of the total spin $S_n:=\sum_{i=1}^n\sigma_i$.
Since the density of the Gibbs measure is symmetric and in particular rotationally invariant,
each of the spins $\sigma_i$ has a uniform marginal distribution, and $\E\inprod{\sigma_i}{
\sigma_i}=1$ for each $i$ and  $\E\inprod{\sigma_i}{\sigma_j}$ is the same for every pair $i\neq j$.

Following \cite{KM}, the density of $\sigma_1$ with respect to uniform measure on $\s^{N-1}$, conditional on $\{\sigma_j\}_{j\neq 1}$, is  
\[Z_1^{-1}\exp\left[\frac{\beta}{n}\sum_{j\neq 1}\inprod{\theta}{\sigma_j}\right],\]
where  $Z_1=\int_{\s^{N-1}}\exp\left[\frac{\beta}{n}\sum_{j\neq 1}\inprod{\theta}{\sigma_j}
\right]d\mu(\theta)$ is the normalization factor.
If $i\in\{1,\ldots,n\}$ is fixed, then call $\sigma^{(i)}:=\sum_{j\neq i}\sigma_j$.
We use hyperspherical coordinates, \[d \mu(\theta) = \frac{J_N}{A_N}
  d \theta_1 \dots d\theta_{N-1},\] where \[J_N = \left(-1 \right)^{N-1} \prod_{k=2}^{N-1} \sin^{k-1}(\theta_k).\]
  Here $A_N = 2 \pi^\frac{N}{2}/\Gamma{\left(\frac{N}{2}\right)}$, and we also use the notation $\kappa=\frac{\beta|\sigma^{(1)}|}{n}$. Therefore the normalization factor is
    \begin{align*}
    Z_1& = \frac{1}{A_N}\int_{0}^{ \pi} ... \int_{0}^{ \pi} \int_{0}^{2 \pi} e^{\kappa \cos(\theta_{N-1})} J_N d \theta_1...d\theta_{N-2} d\theta_{N-1}\\
    &=\frac{1}{A_N} \int_{0}^{ \pi} ... \int_{0}^{ \pi} \int_{0}^{2 \pi} e^{\kappa \cos(\theta_{N-1})} \sin(\theta_2) \sin^{2}(\theta_3)...\sin^{N-2}(\theta_{N-1}) d \theta_1 ...d\theta_{N-2} d\theta_{N-1} \\
    & = \frac{A_{N-1}}{A_{N}} \int_{0}^{ \pi}  e^{\kappa \cos(\theta_{N-1})}  \sin^{N-2}(\theta_{N-1}) d\theta_{N-1} .
    \end{align*}
The conditional expectation can be calculated using the conditional density as follows:
\begin{align*}\E\left[\sigma_1\big|\{\sigma_j\}_{j\neq 1}\right]&=
\frac{1}{Z_1}
\int_{\s^{N-1}}\theta\exp\left[\frac{\beta}{n}\sum_{j\neq 1}\inprod{\theta}{\sigma_j}
\right]d\mu(\theta)\\&=\frac{1}{Z_1}
\int_{\s^{N-1}}\inprod{\theta}{\frac{
\sigma^{(1)}}{|\sigma^{(1)}|}}\left(\frac{\sigma^{(1)}}{|\sigma^{(1)}|}\right)
\exp\left[\frac{\beta}{n}\inprod{\theta}{\sigma^{(1)}}\right]d\mu(\theta)\\
&=
\left[\frac{1}{A_{N} Z_1}\int_{0}^{ \pi} ... \int_{0}^{ \pi} \int_{0}^{2 \pi} \cos(\theta_{N-1}) e^{\kappa \cos(\theta_{N-1})} J_N d \theta_1 ...d\theta_{N-2} d\theta_{N-1}
\right]\cdot\frac{\sigma^{(1)}}{|\sigma^{(1)}|}\\
&=
\left[\frac{1}{Z_1}\frac{A_{N-1}}{A_{N}}\int_{0}^{\pi} \cos(\theta_{N-1}) e^{\kappa \cos(\theta_{N-1})} \sin^{N-2}(\theta_{N-1}) d\theta_{N-1}
\right]\cdot\frac{\sigma^{(1)}}{|\sigma^{(1)}|}
\\
&=\frac{I_{\frac{N}{2}}(\kappa)}{I_{\frac{N}{2}-1}(\kappa)}\cdot\frac{\sigma^{(1)}}{|\sigma^{(1)}|}.
\end{align*}
Here again $ I_{\frac{N}{2}}$ is the modified Bessel function of the first kind.
A series expansion about zero gives $\frac{I_{\frac{N}{2}}(\kappa)}{I_{\frac{N}{2}-1}(\kappa)}\approx\frac{\kappa}{N}$ for small $\kappa$, hence for
$\beta<N$, we have 
\[\E\left[\sigma_1\big|\{\sigma_j\}_{j\neq 1}\right]\approx\frac{\kappa}{N}=
\frac{\beta\sigma^{(1)}}{Nn}=\frac{\beta}{Nn}\sum_{i\neq 1}\sigma_i,\]
and taking an inner product with $\sigma_2$, taking expectation, and using symmetry we obtain:
\begin{equation*}\begin{split}
\E\inprod{\sigma_1}{\sigma_2}&= \E\left[\E\left[\inprod{\sigma_1}{
\sigma_2}|\{\sigma_j\}_{j\neq 1}\right]\right] \approx
\frac{\beta}{Nn}\E\inprod{\sum_{j\neq1}\sigma_j}{
\sigma_2}=\frac{\beta}{Nn}\left[1+(n-2)\E\inprod{\sigma_1}{
\sigma_2}\right],
\end{split}\end{equation*}
and thus
\begin{equation}\label{prod-expectation}
\E\inprod{\sigma_1}{\sigma_2}\approx\frac{\beta}{Nn-\beta(n-2)}\approx
\frac{\beta}{n(N-\beta)}.\end{equation}
\[\]
Finally,
\[\E |S_n|^2=n\E|\sigma_1|^2+n(n-1)\E\inprod{\sigma_1}{\sigma_2}
\approx\frac{2n}{N-\beta}.\]

\medskip
Theorem \ref{subcrit_limit} is an application of an abstract normal approximation theorem from \cite{Me}, a version of Stein's method of exchangeable pairs \cite{St}. The specific version used on the analogous mean-field Heisenberg model is Theorem 14 in \cite{KM}.

\medskip

We need to construct an exchangeable pair $(W_n,W_n')$ for applying these theorems \cite{KM,Me}. Using Gibbs sampling, we start with a configuration $\sigma$ and construct a new configuration $\sigma'$ that differs at only one site by picking $I$ uniformly at random in $\{1, \dots, n\}$ and replacing the original spin $\sigma_I$ by the new spin $\sigma_I'$.
The total spin of the original configuration is 
$W_n=\sqrt{\frac{N-\beta}{n}}\sum_{i=1}^n\sigma_i$ and the total spin of the new configuration is 
$$W_n'=W_n(\sigma')= W_n-\sqrt{\frac{N-\beta}{n}} \sigma_I+\sqrt{\frac{N-\beta}{n}}\sigma_I'.$$ 
The lemma below gives expressions for the quantities $R,R',$ and $\Lambda$ appearing in the cited theorems \cite{KM,Me}.

\begin{lemma}\label{L:errors} If the exchangeable pair $(W_n,W_n')$ is obtained using the Gibbs sampling construction above, and $f(\kappa)=\frac{I_{\frac{N}{2}}(\kappa)}{I_{\frac{N}{2}-1}(\kappa)}$  and $\Lambda=\left(\frac{1-\frac{\beta}{N}}{n}\right)Id,$ then
\begin{enumerate}
\item \[\E\left[W_n'-W_n\big|\sigma\right]=-\Lambda W_n+R,\]
where 
\[R=-\frac{\beta}{Nn^2}W_n-\frac{\beta^3}{N^2(N+2)n^{2}} W_n+\frac{a}{n^{3/2}}\sum_{i=1}^n\left[f(\kappa)-\frac{\kappa}{N}+\frac{\kappa^3}{N^2(N+2)}
\right]\frac{\sigma^{(i)}}{|\sigma^{(i)}|}
.\]
\item  \[\E\left[(W_n'-W_n)(W_n'-W_n)^T\big|\sigma\right]=2\Lambda + R',\]
 with
\begin{equation*}\begin{split}
R'&=\frac{a^2}{Nn} \left [ \frac {1}{n}\sum_{i=1}^nN\sigma_i\sigma_i^T - Id \right]-\left [ \frac{2 \beta}{Nn^2} W_nW_n^T-\frac { 2a^2 \beta}{Nn^3}\sum_{i=1}^n \sigma_i\sigma_i^T \right]\\&\qquad\qquad+\frac{a^2}{n^2}\sum_{i=1}^n
\left\{\left[\frac{I_{\frac{N}{2}}(\kappa)+\kappa I_{\frac{N}{2}+1}(\kappa)}{\kappa I_{ \frac{N}{2}-1}(\kappa)}-\frac {1}{N}\right]P_i+\left[\frac{I_{\frac{N}{2}}(\kappa)}{c I_{\frac{N}{2}-1}(\kappa)}-\frac{1}{N}\right]
P_i^\perp\right.\\&\qquad\qquad\qquad\qquad\qquad\qquad\qquad\left.-
\left[\frac{I_{\frac{N}{2}}(\kappa)}{ I_{\frac{N}{2}-1}(\kappa)}-\frac{\kappa}{N}\right](r_i\sigma_i^T+\sigma_ir_i^T)
\right\}
.\end{split}\end{equation*}
\end{enumerate}
\end{lemma}

Now we will give the bounds for $R$ and $R'$ calculated as
above. Theorem \ref{subcrit_limit} follows from the theorem
in \cite{Me} and Lemmas \ref{L:errors} and
\ref{L:error_bounds}.  

\begin{lemma}\label{L:error_bounds}
For the exchangable pair $(W_n,W_n')$ which is constructed using Gibbs sampling and $R,R'$ as in the previous
lemma, there is a constant $c_\beta$ such that
\begin{enumerate} 
\item \(\E|R|\le\frac{c_\beta(N)}{n^{3/2}};\)
\item \(\E\|R'\|_{HS}\le\frac{c_\beta(N)}{n^{3/2}};\)
\item \(\E|W_n'-W_n|^3\le\frac{c_\beta(N)}{n^{3/2}}.\)
\end{enumerate}
\end{lemma}
{\bf Proof of Lemma \ref{L:errors}:}
Let us denote $a:=\sqrt{N-\beta}$.

For part (a), first define $f(\kappa)=\frac{I_{\frac{N}{2}}(\kappa)}{I_{\frac{N}{2}-1}(\kappa)}$. Then
\begin{equation*}\begin{split}
\E\left[W_n'-W_n\big|\sigma\right]&=-\frac{a}{n^{3/2}}\sum_{i=1}^n\left[
\sigma_i-\E\left[\sigma_i\big|\{\sigma_j\}_{j\neq i}\right]\right]
\\&=-\frac{1}{n}W_n+\frac{a}{n^{3/2}}\sum_{i=1}^nf(\kappa)\left(\frac{\sigma^{(i)}}{|\sigma^{(i)}|}\right).
\end{split}\end{equation*}

For $\beta<N$, we have $\frac{\beta|\sigma^{(i)}|}{n}=o(1)$
with probability exponentially close to 1.  We therefore use the expansion of $f(c)$ near zero to write  
\begin{equation*}\begin{split}
\E\left[W_n'-W_n\big|\sigma\right]&=-\frac{1}{n}W_n+
\left(\frac{a}{Nn^{5/2}}\sum_{i=1}^n\sum_{j\neq i}
\beta\sigma_j\right)-
\left(\frac{a}{N^2(N+2)n^{9/2}}\sum_{i=1}^n\sum_{j\neq i}
\beta^3 (\sigma_j)^3\right)\\&\qquad\qquad\qquad
+\frac{a}{n^{3/2}}\sum_{i=1}^n\left[f(\kappa)-\frac{\kappa}{N}+\frac{\kappa^3}{N^2(N+2)}
\right]\left(\frac{\sigma^{(i)}}{|\sigma^{(i)}|}\right).
\end{split}\end{equation*}
Since
\begin{equation} \label{first-term}
\frac{a}{Nn^{5/2}}\sum_{i=1}^n\sum_{j\neq i}
\beta\sigma_j=\frac{1}{Nn^2}\sum_{i=1}^n\left(\beta W_n-\frac{a\beta
\sigma_i}{\sqrt{n}}\right)=
\frac{1}{n}\left(\frac{\beta}{N}-\frac{\beta}{Nn}\right)W_n.
\end{equation}
 and
 \begin{equation*}\begin{split}
 \frac{a}{N^2(N+2)n^{9/2}}\sum_{i=1}^n\sum_{j\neq i}
  \beta^3 (\sigma_j)^3 = \frac{a \beta^3}{N^2(N+2)n^{9/2}}\left((n-1)+(n-2)(n-1)\E\inprod{\sigma_1}{
  \sigma_2}\right) \sum_{i=1}^n\sum_{j\neq i}
  \sigma_j.
 \end{split}\end{equation*}
 Rewriting the right hand side of the last equation:
 \begin{equation*}\begin{split}
  \left[\frac{a \beta}{Nn^{5/2}}\sum_{i=1}^n\sum_{j\neq i}
      \sigma_j\right]\left[\frac{ \beta^2}{N(N+2)n^{2}}\left((n-1)+(n-2)(n-1)\E\inprod{\sigma_1}{
            \sigma_2}\right) \right].
  \end{split}\end{equation*}
  Using \eqref{prod-expectation} and \eqref{first-term} the last equation leads us to
  \begin{equation*}\begin{split}
   \frac{a}{N^2(N+2)n^{9/2}}\sum_{i=1}^n\sum_{j\neq i}
    \beta^3 (\sigma_j)^3 \approx \frac{\beta^3}{N^2(N+2)n^{2}} \left(W_n-\frac{W_n}{n}\right)
   \end{split}.\end{equation*}
Therefore, 
\begin{equation*}\begin{split}
\E\left[W_n'-W_n\big|\sigma\right]&=-\frac{1}{n}W_n+
\frac{W_n}{n}\left(\frac{\beta}{N}-\frac{\beta}{Nn}\right)-\frac{\beta^3}{N^2(N+2)n^{2}} W_n\\&\qquad\qquad\qquad
+\frac{a}{n^{3/2}}\sum_{i=1}^n\left[f(\kappa)-\frac{\kappa}{N}+\frac{\kappa^3}{N^2(N+2)}
\right]\left(\frac{\sigma^{(i)}}{|\sigma^{(i)}|}\right).
\end{split}\end{equation*}

\begin{equation*}\begin{split}
\E\left[W_n'-W_n\big|\sigma\right]&=-\left(\frac{1}{n}-\frac{\beta}{Nn}\right)W_n-\frac{\beta}{Nn^2}W_n-\frac{\beta^3}{N^2(N+2)n^{2}} W_n\\&\qquad\qquad\qquad+\frac{a}{n^{3/2}}\sum_{i=1}^n\left[f(\kappa)-\frac{\kappa}{N}+\frac{\kappa^3}{N^2(N+2)}
\right]\left(\frac{\sigma^{(i)}}{|\sigma^{(i)}|}\right).
\end{split}\end{equation*}

The matrix $\Lambda$ of the theorem from \cite{Me} is thus
$\frac{N-\beta}{Nn}Id$ and 
\[R=-\frac{\beta}{Nn^2}W_n-\frac{\beta^3}{N^2(N+2)n^{2}} W_n+\frac{a}{n^{3/2}}\sum_{i=1}^n\left[f(\kappa)-\frac{\kappa}{N}+\frac{\kappa^3}{N^2(N+2)}
\right]\left(\frac{\sigma^{(i)}}{|\sigma^{(i)}|}\right)
.\]

This completes the proof of part(a).\\
For part (b), similarly we have, 
\begin{equation*}\begin{split}
\E&\left[(W_n'-W_n)(W_n'-W_n)^T\big|\sigma\right]\\&\qquad=\frac{a^2}{n^2}\sum_{i=1}^n
\frac{1}{Z_i}\int_{\s^{N-1}}(\theta-
\sigma_i)(\theta-\sigma_i)^T\exp\left[\frac{\beta}{n}
\sum_{j\neq i}\inprod{\sigma_j}{
\theta}\right]d\mu(\theta)\\&\qquad=\frac{a^2}{n^2}\sum_{i=1}^n
\frac{1}{Z_i}\int_{\s^{N-1}}\left[\theta\theta^T-\sigma_i\theta^T-\theta\sigma_i^T
+\sigma_i\sigma_i^T\right]\exp\left[\frac{\beta}{n}\sum_{j\neq i}
\inprod{\sigma_j}{\theta}\right]d\mu(\theta).
\end{split}\end{equation*}
Due to symmetry we have, $Z_i = Z_1$ for all i. Now letting $\theta=\theta_1+\theta_2$, where $\theta_1$ is the projection 
of $\theta$ onto the direction $\sigma^{(i)}$, the first term of the $i^{th}$ summand is 
\begin{equation}\label{theta_part}
\frac{1}{Z_1}\int_{\s^{N-1}}\big[\theta\theta^T\big]
\exp\left[\frac{\beta}{n}\inprod{\sigma^{(i)}}{
\theta}\right]d\mu(\theta)=\frac{1}{Z_1}\int_{\s^{N-1}}\big[\theta_1\theta_1^T+\theta_2\theta_2^T\big]
\exp\left[\frac{\beta}{n}\inprod{\sigma^{(i)}}{
\theta}\right]d\mu(\theta).\end{equation} 
To compute it, write $r_i = \frac{
\sigma^{(i)}}{|\sigma^{(i)}|}$, this implies $\theta_1 =
  \inprod{\theta}{r_i} r_i,$ and $ \theta_1\theta_1^T = \left|
  \inprod{\theta}{r_i} \right|^2r_ir_i^T.$ Define $c :=
  \frac{\beta|\sigma^{(i)}|}{n},$ 
  \begin{equation*}\begin{split}
  \frac{1}{Z_1}\int_{\s^{N-1}} \theta_1\theta_1^T  \exp\left[\frac{\beta}{n}
  \inprod{\sigma^{(i)}}{\theta}\right]d\mu(\theta) & = \frac{1}{Z_1}
  \left( \int_{\s^{N-1}} \left|\inprod{\theta}{r_i} \right|^2 \exp\left[ c
    \inprod{r_i}{\theta}\right]d\mu(\theta)\right)r_ir_i^T\\& =\frac{1}{Z_1} \frac{A_{N-1}}{A_{N}}
     \left[\int_{0}^{\pi} \left(\cos(\theta_{N-1})\right)^2e^{\kappa \cos(\theta_{N-1})} \sin^{N-2}(\theta_{N-1}) d\theta_{N-1}
     \right] r_ir_i^T 
     \\& =\frac{1}{Z_1} \frac{A_{N-1}}{A_{N}}
          \left[\int_{-1}^{1} u^2e^{\kappa u} \left(1-u^2\right)^\frac{N-3}{2} du
          \right] r_ir_i^T .
  \end{split}\end{equation*}
Using the definition of $Z_i$ we obtain
\begin{align*}
\frac{1}{Z_1}\int_{\s^{N-1}} \theta_1\theta_1^T  \exp\left[\frac{\beta}{n}
\inprod{\sigma^{(i)}}{\theta}\right]d\mu(\theta)&=\frac{\int_{-1}^{1} u^2e^{\kappa u} \left(1-u^2\right)^\frac{N-3}{2} du
          }{\int_{-1}^{1} e^{\kappa u} \left(1-u^2\right)^\frac{N-3}{2} du
                    } \\& =
\frac{I_{\frac{N}{2}}(\kappa)+\kappa I_{\frac{N}{2}+1}(\kappa)}{\kappa I_{ \frac{N}{2}-1}(\kappa)}P_i ,
\end{align*}
where $P_i$ is the orthogonal projection onto $r_i$.
Let $\theta_2$ = $(x_1,x_2,...,x_{N-1})$ be the orthonormal coordinate representation for $\theta_2$ within $r_1^\perp$.

Note that using symmetry, $ i \neq j$ we have
\[\int_{\s^{N-1}}x_{i}x_{j}e^{c\inprod{r_i}{\theta}}d\mu(\theta)=0\]
A polar coordinate expansion yields:
\begin{equation*}\begin{split}
& \int_{\s^{N-1}}x_{N-1}^2 e^{c\inprod{r_i}{\theta}}d\mu(\theta)\\&= \frac{1}{A_{N}}\int_{0}^{ \pi} ... \int_{0}^{2 \pi}\left(\cos(\theta_{N-2}) \sin(\theta_{N-1})\right)^2 e^{\kappa \cos(\theta_{N-1})}J_N d \theta_1 ...d\theta_{N-1}\\&= \frac{A_{N-2}}{A_{N}} \int_{0}^{ \pi} \int_{0}^{ \pi} \left(\cos(\theta_{N-2}) \sin(\theta_{N-1})\right)^2 e^{\kappa \cos(\theta_{N-1})} \sin^{N-3}(\theta_{N-2}) \sin^{N-2}(\theta_{N-1}) d\theta_{N-2} d\theta_{N-1} \\& =\frac{A_{N-2}}{A_{N}} \frac{2 \prod_{j=1}^{\frac{N}{2}-2} \left(2j\right)}{\prod_{j=0}^{\frac{N}{2}-1} \left(2j+1\right)} \int_{0}^{ \pi} e^{\kappa \cos(\theta_{N-1})} \sin^{N}(\theta_{N-1}) d\theta_{N-1} 
\\& =\frac{A_{N-2}}{A_{N}} \frac{2 \prod_{j=1}^{\frac{N}{2}-2} \left(2j\right)}{\prod_{j=0}^{\frac{N}{2}-1} \left(2j+1\right)}
          \left[\int_{-1}^{1} e^{\kappa u} \left(1-u^2\right)^\frac{N-1}{2} du
          \right]. 
\end{split}\end{equation*}

We notice by a series expansion that for $1 \le j \le N-1$, the values of $\int_{\s^{N-1}}x_{j}^2 e^{c\inprod{r_i}{\theta}}d\mu(\theta)$ are the same for very small $\kappa$. Therefore,
\begin{equation*}
\begin{split}
\frac{1}{Z_1}\int_{\s^{N-1}}\theta_2\theta_2^T\exp\left[\frac{\beta}{n}
\inprod{\sigma^{(i)}}{\theta}\right]d\mu(\theta)&=\left[\frac{I_{\frac{N}{2}}(\kappa)}{ \kappa I_{\frac{N}{2}-1}(\kappa)}\right]P_i^\perp  ,\end{split}
\end{equation*}
where $P_i^\perp$ is the orthogonal projection onto $r_1^\perp$.

Similarly we can obtain,
\begin{equation*}\begin{split}
\frac{1}{Z_1}\int_{\s^{N-1}}\theta\sigma_i^T\exp\left[\frac{\beta}{n}\inprod{\sigma^{(i)}}{
\theta}\right]d\mu(\theta)&=\left[\frac{I_{\frac{N}{2}}(\kappa)}{ I_{\frac{N}{2}-1}(\kappa)}\right]r_i\sigma_i^T .
\end{split}\end{equation*}

\begin{equation*}\begin{split}
\frac{1}{Z_1}\int_{\s^{N-1}}\theta ^T\sigma_i\exp\left[\frac{\beta}{n}\inprod{\sigma^{(i)}}{
\theta}\right]d\mu(\theta)&=\left[\frac{I_{\frac{N}{2}}(\kappa)}{ I_{\frac{N}{2}-1}(\kappa)}\right]r_i\sigma_i^T .
\end{split}\end{equation*}

\[\frac{1}{Z_1}\int_{\s^{N-1}}\sigma_i\sigma_i^T\exp\left[\frac{\beta}{n}
\inprod{\sigma^{(i)}}{\theta}\right]d\mu(\theta)
=\sigma_i\sigma_i^T.\]

Collecting all terms we obtain:

\begin{equation*}\begin{split}
\E&\left[(W_n'-W_n)(W_n'-W_n)^T\big|\sigma\right]
\\&\qquad=\frac{a^2}{n^2}\sum_{i=1}^n\left\{\left[\frac{I_{\frac{N}{2}}(\kappa)+\kappa I_{\frac{N}{2}+1}(\kappa)}{\kappa I_{ \frac{N}{2}-1}(\kappa)}\right]P_i+\left[\frac{I_{\frac{N}{2}}(\kappa)}{\kappa I_{\frac{N}{2}-1}(\kappa)}\right]P_i^\perp-\left[\frac{I_{\frac{N}{2}}(\kappa)}{ I_{\frac{N}{2}-1}(\kappa)}\right]
(r_i\sigma_i^T+\sigma_ir_i^T)+\sigma_i\sigma_i^T\right\}.
\end{split}\end{equation*}

Remembering that $\kappa=\frac{\beta|\sigma^{(i)}|}{n}$, we have
\begin{equation}\begin{split}\label{quad-diff1}
\E\left[(W_n'-W_n)(W_n'-W_n)^T\big|\sigma\right]&=
\frac{a^2}{n^2}\sum_{i=1}^n\left\{\frac{1}{N}P_i+\frac{1}{N}P_i^\perp-\frac{\kappa}{N}
(r_i\sigma_i^T+\sigma_ir_i^T)+\sigma_i\sigma_i^T\right\}+R''\\
&=\frac{a^2}{Nn}Id+\frac{a^2}{Nn^2}\sum_{i=1}^n
N \sigma_i\sigma_i^T-\frac{a^2 \kappa}{Nn^2}\sum_{i=1}^n(r_i\sigma_i^T+\sigma_ir_i^T)
+R'',
\end{split}\end{equation}
where the remainder term $R''$ is given by
\begin{equation*}\begin{split}
R''&=\frac{a^2}{n^2}\sum_{i=1}^n
\left\{\left[\frac{I_{\frac{N}{2}}(\kappa)+\kappa I_{\frac{N}{2}+1}(\kappa)}{\kappa I_{ \frac{N}{2}-1}(\kappa)}-\frac {1}{N}\right]P_i+\left[\frac{I_{\frac{N}{2}}(\kappa)}{c I_{\frac{N}{2}-1}(\kappa)}-\frac{1}{N}\right]
P_i^\perp\right.\\&\qquad\qquad\qquad\qquad\qquad\qquad\qquad\left.-
\left[\frac{I_{\frac{N}{2}}(\kappa)}{ I_{\frac{N}{2}-1}(\kappa)}-\frac{\kappa}{N}\right](r_i\sigma_i^T+\sigma_ir_i^T)
\right\}.\end{split}\end{equation*}
Using $r_i=\frac{\sigma^{(i)}}{|\sigma^{(i)}|}$ and $\kappa=\frac{\beta| \sigma^{(i)}|}{n}$, the third term of \eqref{quad-diff1} simplifies:
\begin{equation*}\begin{split}
\frac{a^2 \kappa}{Nn^2}\sum_{i=1}^n(r_i\sigma_i^T+\sigma_ir_i^T)&=
\frac{a^2\beta}{Nn^3}\sum_{i=1}^n\sum_{j\neq i}(\sigma_j\sigma_i^T+\sigma_i
\sigma_j^T)=\frac{2\beta}{Nn^2}W_nW_n^T-\frac{2a^2\beta}{Nn^3}\sum_{i=1}^n
\sigma_i\sigma_i^T.
\end{split}\end{equation*}
Collecting all terms,
\begin{equation*}\begin{split}
\E&\left[(W_n'-W_n)(W_n'-W_n)^T\big|\sigma\right]
\\&\qquad=\left(\frac{N-\beta}{Nn}\right)\left[Id+\frac{1}{n}\sum_{i=1}^n
N\sigma_i\sigma_i^T\right]-\frac{2 \beta}{Nn^2} W_nW_n^T+\frac{2 a^2 \beta}{Nn^3}\sum_{i=1}^n\sigma_i\sigma_i^T+R''+\frac{a^2}{Nn}Id-\frac{a^2}{Nn}Id
\\&\qquad= 2 \left (\frac{N-\beta}{Nn}\right)Id +R'''\\&\qquad=  2\Lambda + R''',
\end{split}\end{equation*}
where 

\begin{equation*}\begin{split}
R'''&=\frac{a^2}{Nn} \left [ \frac {1}{n}\sum_{i=1}^nN\sigma_i\sigma_i^T - Id \right]-\left [ \frac{2 \beta}{Nn^2} W_nW_n^T-\frac { 2a^2 \beta}{Nn^3}\sum_{i=1}^n \sigma_i\sigma_i^T \right]+R''
\end{split}\end{equation*}\qed

{\bf Proof of Lemma \ref{L:error_bounds}:}
\\ Recall that $\Lambda=\left(\frac{N-\beta}{Nn}\right)Id$, and thus $\|\Lambda^{-1}\|_{op}=\frac{Nn}{N-\beta}.$  Now, from Lemma \ref{L:errors},
\[R=-\frac{\beta}{Nn^2}W_n-\frac{\beta^3}{N^2(N+2)n^{2}} W_n+\frac{a}{n^{3/2}}\sum_{i=1}^n\left[f(\kappa)-\frac{\kappa}{N}+\frac{\kappa^3}{N^2(N+2)}
\right]\left(\frac{\sigma^{(i)}}{|\sigma^{(i)}|}\right)
.\]
From our earlier heuristic approach,  $\E|W_n|^2\approx N$. We can use the same argument together with the fact
that $\frac{I_{\frac{N}{2}}(\kappa)}{ I_{\frac{N}{2}-1}(\kappa)}\le\frac{\kappa}{N}-\frac{\kappa^3}{N^2(N+2)}$ to prove that $\E|W_n|^2\le N$. Using this condition we can bound the first two terms on R.H.S. of $R$ as follows:

\[ \E\left[-\frac{\beta}{Nn^2}W_n-\frac{\beta^3}{N^2(N+2)n^{2}} W_n\right]  \leq \left(-\frac{\beta}{N}-\frac{\beta^3}{N^2(N+2)} \right) \frac{\sqrt{N}}{n^2}\] 

 For third term estimation of $R$, fix $\epsilon=\epsilon(n)\in(0,1)$ 
which will be defined later.  Define
$\widetilde{r}(\kappa):=\frac{I_{\frac{N}{2}}(\kappa)}{ I_{\frac{N}{2}-1}(\kappa)}-\frac{\kappa}{N}+\frac{\kappa^3}{N^2(N+2)};$ observe that for $b$ a universal constant, if $t\le\epsilon$
then $|\widetilde{r}(t)|<b\epsilon^4$. Therefore
\begin{equation}\begin{split}\label{Ept2}
\frac{a}{n^{3/2}}\left|
\sum_{i=1}^n\left[\widetilde{r}\left(\frac{\beta|\sigma^{(i)}|}{n}\right)
\right]\left(\frac{\sigma^{(i)}}{|\sigma^{(i)}|}\right)\right|\le\frac{ba\epsilon^4}{\sqrt{n}}+\frac{a}{n^{3/2}}\sum_{i=1}^n\1\left(\frac{
\beta|\sigma^{(i)}|}{n}>\epsilon\right),
\end{split}\end{equation}
where we used $\big|\widetilde{r}\left(\frac{\beta|\sigma^{(i)}|}{n}\right)\big|\le1$
for any configuration $\sigma$.
From an adaptation of proposition \ref {SpinLDPM}, since $M_n = \frac{\sigma^{(i)}+\sigma_i}{n}$  the LDP for $M_n$, we deduce 
\[\P\left[\frac{\beta|\sigma^{(i)}|}{n}>\epsilon\right]\le
C\exp\left[-\frac{n}{2}\inf\{I_\beta(r):r\ge\epsilon\}\right],\]
where for $r = |x|$, from proposition \ref {SpinLDPM} and (\ref{tbm2})  we have 

\[I_\beta(r) =r \frac {I_\frac{N}{2}(r)}{I_{\frac{N}{2}-1}(r)}+\log\left[\frac{A_{N}}{A_{N-1}} \frac {r^{\frac{N}{2}-1}}{B_N\pi I_{\frac{N}{2}-1}(r) }\right]-\frac{\beta}{2}\left(\frac {I_\frac{N}{2}(r)}{I_{\frac{N}{2}-1}(r)}\right)^2, \]
with
\[B_N= \begin{cases} \prod_{k=0}^{\frac{N}{2}-1} |2k-1| , \quad \quad \quad \quad \quad \quad \quad \quad \text{ if N even}, \\ \\ 2^{\frac{N-3}{2}} (1)_{\frac{N-3}{2}},  \quad \quad \quad \quad \quad \quad \quad \quad \quad\;\, \text{ if N odd}, \end{cases}\]
where $(a)_n$ represents pochhammer symbol. Using Taylor series expansion, we can deduce that there is a universal constant $q>0$ such that for $\epsilon\in(0,1)$,
$I_\beta(\epsilon)\ge\frac{\epsilon^2}{2N}\left(1-\frac{\beta}{N}\right)-q \epsilon^3$.
It follows that
\[\P\left[\frac{\beta|\sigma^{(i)}|}{n}>\epsilon\right]\le
C\exp\left[-\frac{n}{2}\left(\frac{\epsilon^2}{2N}\left(1-\frac{\beta}{N}\right)-q \epsilon^3\right)  \right].\]
Choose $\epsilon=\epsilon(n)$ such that
$\epsilon^2=\frac{4N\log(n)}{n \left(1-\frac{\beta}{N}\right)}$.
Then
\(\P\left[\frac{\beta|\sigma^{(i)}|}{n}>\epsilon\right]\le
\frac{C'}{n},\) from the bound in \eqref{Ept2} we notice that the second term is bounded by $n^{-3/2}$.
Now we would be interested in bounding the first term. Notice that 

\[\lim_{n\to\infty} n \epsilon^4   = \lim_{n\to\infty} \frac{16N^2\log(n)^2}{n \left(1-\frac{\beta}{N}\right)^2} = 0\]
which implies that $\epsilon^4$ is bounded above by $\frac{1}{n}$. This leads us to the conclusion that the first term of \eqref{Ept2} is also bounded by $n^{-3/2}$. Therefore,
\begin{equation*}\begin{split}
\frac{a}{n^{3/2}}\E\left|
\sum_{i=1}^n\left[\widetilde{r}\left(\frac{\beta|\sigma^{(i)}|}{n}\right)
\right]\left(\frac{\sigma^{(i)}}{|\sigma^{(i)}|}\right)\right|&\,\le\,
\frac{ba\epsilon^4}{\sqrt{n}}+\frac{a}{\sqrt{n}}\P\left[\frac{\beta|\sigma^{(1)}|}{n}>\epsilon\right]\,\le\,\frac{c_\beta(N) }{n^{3/2}}.
\end{split}\end{equation*}
This completes the proof of part (a).

\medskip

For part (b), the value of $R'$ from Lemma \ref{L:errors} is given by
\begin{equation*}\begin{split}
R'&=\frac{a^2}{Nn} \left [ \frac {1}{n}\sum_{i=1}^nN\sigma_i\sigma_i^T - Id \right]-\left [ \frac{2 \beta}{Nn^2} W_nW_n^T-\frac { 2a^2 \beta}{Nn^3}\sum_{i=1}^n \sigma_i\sigma_i^T \right]\\&\qquad\qquad+\frac{a^2}{n^2}\sum_{i=1}^n
\left\{\left[\frac{I_{\frac{N}{2}}(\kappa)+\kappa I_{\frac{N}{2}+1}(\kappa)}{\kappa I_{ \frac{N}{2}-1}(\kappa)}-\frac {1}{N}\right]P_i+\left[\frac{I_{\frac{N}{2}}(\kappa)}{c I_{\frac{N}{2}-1}(\kappa)}-\frac{1}{N}\right]
P_i^\perp\right.\\&\qquad\qquad\qquad\qquad\qquad\qquad\qquad\left.-
\left[\frac{I_{\frac{N}{2}}(\kappa)}{ I_{\frac{N}{2}-1}(\kappa)}-\frac{\kappa}{N}\right](r_i\sigma_i^T+\sigma_ir_i^T)
\right\}
.\end{split}\end{equation*}
For  $x\in\R^n$,
 $\|xx^T\|_{HS}=|x|^2$, and thus $\E\|\sigma_i \sigma_i^T\|_{HS}=\E|\sigma_i|^2=1$, also
\[\E\|W_nW_n^T\|_{HS}=\E|W_n|^2\le N.\]
For estimating $\E\|\frac{1}{n}\sum_{i=1}^n(N\sigma_i\sigma_i^T-Id)\|_{HS}$, recall that $\|A\|_{HS}=\sqrt{\tr(AA^T)},$
and then by the Cauchy-Schwarz inequality,
\[\E\left\|\frac{1}{n}\sum_{i=1}^n(N\sigma_i\sigma_i^T-Id)\right\|_{HS}\le
\frac{1}{n}\sqrt{\sum_{i,j=1}^n\E\tr\big[(N\sigma_i\sigma_i^T-Id)(N\sigma_j
\sigma_j^T-Id)\big]}.\]
So,
\begin{equation*}\begin{split}
\E\tr\big[(N\sigma_i\sigma_i^T-Id)^2\big]
=N^2\E|\sigma_i|^4-2N\E|\sigma_i|^2+2=N^2-2N+N.
\end{split}\end{equation*}
Similarly, for $i\neq j$,
\begin{equation*}\begin{split}
\E\tr\big[(N\sigma_i\sigma_i^T-Id)(N\sigma_j
\sigma_j^T-Id)\big]&=N^2\E\left[\inprod{\sigma_i}{\sigma_j}^2\right]-
N.
\end{split}\end{equation*}
Notice that
\begin{equation*}\begin{split}
\E\left[\inprod{\sigma_1}{\sigma_2}^2\big|\{\sigma_i\}_{i\neq 1}\right]&=
\sigma_2^T\E\left[\sigma_1\sigma_1^T\big|\{\sigma_i\}_{i\neq 1}\right]\sigma_2\\
&=\sigma_2^T\left(\frac{1}{Z_1}\int_{\s^{N-1}}\theta\theta^T\exp\left[\frac{\beta}{n}
\inprod{\theta}{\sigma^{(1)}}\right]d\mu(\theta)\right)\sigma_2\\
&=\sigma_2^T\left(\left[\frac{I_{\frac{N}{2}}(\kappa)+\kappa I_{\frac{N}{2}+1}(\kappa)}{\kappa I_{ \frac{N}{2}-1}(\kappa)}\right]P_i+\left[\frac{I_{\frac{N}{2}}(\kappa)}{\kappa I_{\frac{N}{2}-1}(\kappa)}\right]P_i^\perp\right)\sigma_2,
\end{split}\end{equation*}
Again since
$\kappa=o(1)$ with high probability, $\frac{I_{\frac{N}{2}}(\kappa)}{ I_{\frac{N}{2}-1}(\kappa)}\le\frac{\kappa}{N}-\frac{\kappa^3}{N^2(N+2)}$ and
\begin{equation*}\begin{split}
\E\left[\inprod{\sigma_1}{\sigma_2}^2\big|\{\sigma_i\}_{i\neq 1}\right]
\approx\sigma_2^T\left(\left(\frac{1}{N}+\frac{(N-1)\kappa^2}{N^2(N+2)}\right)P_i+\left(\frac{1}{N}-\frac{\kappa^2}{N^2(N+2)}\right)P_i^\perp\right)\sigma_2, 
\end{split}\end{equation*}

\begin{equation*}\begin{split}
\E\left[\inprod{\sigma_1}{\sigma_2}^2\big|\{\sigma_i\}_{i\neq 1}\right]
\approx 
\frac{1}{N}-\frac{\kappa^2}{N^2(N+2)}+\frac{N\kappa^2}{N^2(N+2)}
\E\left[\sigma_2^T\sigma_1\sigma_1^T\sigma_2\big|\{\sigma_i\}_{i\neq 1}
\right],
\end{split}\end{equation*}

\[\E\left[\inprod{\sigma_1}{\sigma_2}^2\right]\approx
\frac{1}{N}+\E\left[ -\frac{\kappa^2}{N^2(N+2)}+\frac{N\kappa^2}{N^2(N+2)} \inprod{\sigma_1}
{\sigma_2}^2\right].\]
This implies
\[\E\tr\big[(N\sigma_i\sigma_i^T-Id)(N\sigma_j\sigma_j^T-Id)\big]\approx
-\frac{\kappa^2}{(N+2)} + \frac{N\kappa^2}{(N+2)}\E\left[\inprod{\sigma_1}
{\sigma_2}^2\right].\]
For each pair $i\neq j$ and $r_1$ being a universal constant, the error in this approximation is represented by 

\begin{equation*}\begin{split}
\E\left|N^2\sigma_2^T \left(\left[\frac{I_{\frac{N}{2}}(\kappa)+\kappa I_{\frac{N}{2}+1}(\kappa)}{\kappa I_{ \frac{N}{2}-1}(\kappa)}-\frac{1}{N}-\frac{(N-1)\kappa^2}{N^2(N+2)}\right]P_i+\left[\frac{I_{\frac{N}{2}}(\kappa)}{\kappa I_{\frac{N}{2}-1}(\kappa)}-\frac{1}{N}+\frac{\kappa^2}{N^2(N+2)}\right]P_i^\perp\right)\sigma_2\right| 
\end{split}\end{equation*}
and it is bounded by $ r_1\E \kappa^3$.
Now, 
\[\E \kappa^2=\frac{\beta^2}{n^2}\sum_{i,j>1}\E\inprod{\sigma_i}{\sigma_j}\le
\frac{\beta^2}{n^2}\left[(n-1)+(n-1)(n-2)\frac{\beta}{n(N-\beta)}\right]\le
\frac{N\beta^2}{n(N-\beta)},\]
and
\[\E c^3\le\frac{N\beta^3}{n(N-\beta)},\]
therefore
\[\E\left\|\frac{1}{n}\sum_{i=1}^n(N\sigma_i\sigma_i^T-Id)\right\|_{HS}\le
\sqrt{\frac{(N^2-2N+N)+\frac{r_2(\beta^2+\beta^3)}{(N-\beta)}}{n}}.\]

Using Taylor expansion for $R''$ (which is remaining term of the error $R'$ in Lemma \ref{L:errors}) and for a universal constant $c_*$,we obtain:
\begin{equation*}\begin{split}
\E\|R''\|_{HS}&\le\frac{a^2}{n^2}\sum_{i=1}^n\E
\left\{\left\|\left[\frac{I_{\frac{N}{2}}(\kappa)+\kappa I_{\frac{N}{2}+1}(\kappa)}{\kappa I_{ \frac{N}{2}-1}(\kappa)}-\frac{1}{N}
\right]P_i\right\|_{HS}+\left\|\left[\frac{I_{\frac{N}{2}}(\kappa)}{\kappa I_{\frac{N}{2}-1}(\kappa)}-\frac{1}{N}\right]
P_i^\perp\right\|_{HS}\right.\\&\qquad\qquad\qquad\qquad\qquad\qquad\qquad\left.-
\left\|\left[\frac{I_{\frac{N}{2}}(\kappa)}{ I_{\frac{N}{2}-1}(\kappa)}-\frac{\kappa}{N}\right](r_i\sigma_i^T+\sigma_ir_i^T)
\right\|_{HS}\right\}\\
&\le\frac{c_*(N-\beta)\beta}{n^3}\sum_{i=1}^n\E|\sigma^{(i)}|\\
&\le\frac{c_*\sqrt{N(N-\beta)}\beta}{n^{3/2}} \le\frac{c_\beta(N)}{n^{3/2}} ,\end{split}\end{equation*}
 where we used the facts that $\|P_i\|_{HS}$, $\|P_i^\perp\|_{HS}$ and $\|r_i\sigma_i^T\|_{HS}$
are all bounded by $\sqrt{N}$ or smaller and that $\E|\sigma^{(i)}|\le\sqrt{\frac{Nn}{
N-\beta}}$.This completes the proof of part (b).

Finally, part (c) is trivial and identical to \cite{KM}
with different $\sigma$ belonging to the sphere.
\qed

\section{The Total Spin in the supercritical Phase}\label{S:supcrit}

For proving theorem \ref{T:supcrit_CLT}, we will use a version of Stein's abstract normal approximation theorem  \cite{St} (p.35).  The formulation given below is a univariate analog of abstract normal approximation theorem from \cite{Me}.

Consider the random variable $W_n=\sqrt{n}\left[\frac{\beta^2}{n^2b^2}\left|\sum_{j=1}^n\sigma_j
\right|^2-1\right]$, for supercritical case as explained in section 3. Now construct an exchangable pair $(W_n,W_n')$ using Glauber dynamics in order to apply Stein's abstract normal approximation theorem to $W_n$.
We will describe the lemma which contains the bounds needed to obtain Theorem
\ref{T:supcrit_CLT} from Stein abstract theorem. This will lead us to
proof of Theorem \ref{T:supcrit_CLT}.
\begin{lemma}
Define $f(r)=\frac{I_{\frac{N}{2}}(r)}{I_{\frac{N}{2}-1}(r)}$, and let $b$ be the positive solution of \[b-\beta f(b)=0. \] Then for the exchangable pair $(W_n,W_n')$ as constructed above, 
\begin{enumerate}
\item For $\lambda=\frac{1-\beta  f'(b)}{n}$, 
\[\E\left[W_n'-W_n\big|\sigma\right] =-\lambda W_n+R\qquad
and\qquad
\E|R|\le\frac{c_\beta(N)\log(n)}{n^{3/2}};\]
\item For $Var(\sigma)=\frac{4\beta^2}{\left(1-\beta f'(b)\right)b^2}
\left[1-\frac{N-1}{N}\frac{\left(I_{\frac{N}{2}-1}(b)+I_{\frac{N}{2}+1}(b)
\right)}{I_{\frac{N}{2}-1}(b)}-\left(\frac{I_{\frac{N}{2}}(b)}{I_{\frac{N}{2}-1}(b)}\right)^2\right],$
\[\E\left|\sigma^2-\frac{1}{2\lambda}
\E\left[(W_n'-W_n)^2\big|\sigma\right]\right|\le\frac{c_\beta(N)(\log(n))^{1/4}}{n^{1/4}};\]
\item $\E|W_n'-W_n|^3\le\frac{c_\beta}{n^{3/2}}.$
\end{enumerate}
\end{lemma}

\begin{proof}
First of all  $\lambda$ and $\sigma$ defined above are always strictly positive. For par (a), first we found the bounds for $f(x)$ and then using LDP for $|S_n|$ along with taylor series expansion we can deduce the result. Consider,

\begin{equation}\begin{split}\label{E:lindiff1}
\E\left[W_n'-W_n\big|\sigma\right]
&=-\frac{\beta^2}{n^{5/2}b^2}\sum_{i=1}^n\left[
2\sum_{b\neq i}\inprod{\sigma_i}{\sigma_k}-\E\left[2\sum_{k\neq i}\inprod{\sigma_i 
}{\sigma_k}\big|\{\sigma_j\}_{j\neq i}\right]\right]
\\&=-\frac{2\beta^2}{n^{5/2}b^2}\left(\left|\sum_{i=1}^n\sigma_i\right|^2-n\right)+\frac{2\beta^2}{n^{5/2}b^2}\sum_{i=1}^nf\left(\frac{\beta
|\sigma^{(i)}|}{n}\right)
|\sigma^{(i)}|\\&=-\frac{2}{n}W_n-\frac{2}{\sqrt{n}}+\frac{2\beta^2}{n^{3/2}b^2}+
\frac{2\beta^2}{n^{5/2}b^2}\sum_{i=1}^nf\left(\frac{\beta
|\sigma^{(i)}|}{n}\right)
|\sigma^{(i)}|. 
\end{split}\end{equation}
For $f(r)$, using the same approach as in section 5 of \cite{KM}, we have
 
\begin{equation*}\begin{split}
\frac{2\beta^2}{n^{5/2}b^2}\sum_{i=1}^n&\E\left|f\left(\frac{\beta
|\sigma^{(i)}|}{n}\right)|\sigma^{(i)}|-f\left(\frac{\beta
|S_n|}{n}\right)|S_n|\right|\\&\le\frac{2\|f'\|_\infty\beta^3+2\|f\|_\infty\beta^2}{n^{3/2}b^2};\end{split}\end{equation*}
that is,
\begin{equation}\begin{split}\label{E:first-improvement}
\E\left[W_n'-W_n\big|\sigma\right]&=-\frac{2}{n}W_n-\frac{2}{\sqrt{n}}+
\frac{2\beta^2}{n^{3/2}b^2}f\left(\frac{\beta|S_n|}{n}\right)|S_n|+R_1, 
\end{split}\end{equation}
where $\E|R_1|\le\frac{c_\beta}{n^{3/2}}.$

Now we will use Taylor expansion to approximate $f\left(\frac{\beta|S_n|}{n}\right)$
using the LDP for $|S_n|$ (Proposition \ref{SpinLDPM}). For $r = |x|$, we obtain

\begin{equation}\label{norm_LDP}
\limsup_{n\to\infty}\frac{1}{n}P_{n,\beta}\left[\left|
\frac{|S_n|}{n}-\frac{b}{\beta}\right|\ge\epsilon\right]\le
-\inf_{\left|r-\frac{b}{\beta}\right|\ge\epsilon} I_\beta(r),\end{equation}
where
\[I_\beta(r)=\Phi_\beta(y)-\varphi(\beta),\]
with
\[\Phi_\beta(y) =r~y+\log\left[\frac{A_{N}}{A_{N-1}} \frac {y^{\frac{N}{2}-1}}{B_N\pi I_{\frac{N}{2}-1}(y) }\right]-\frac{\beta}{2}r^2 ,\]
\[B_N= \begin{cases} \prod_{k=0}^{\frac{N}{2}-1} |2k-1| , \quad \quad \quad \quad \quad \quad \quad \quad \text{ if N even}, \\ \\ 2^{\frac{N-3}{2}} (1)_{\frac{N-3}{2}},  \quad \quad \quad \quad \quad \quad \quad \quad \quad\;\, \text{ if N odd}, \end{cases}\]
here $y$ is calculated from 
\[\frac {I_\frac{N}{2}(y)}{I_{\frac{N}{2}-1}(y)} = r,\] and \[\varphi(\beta) = \inf_{ x \ge 0} \Phi_\beta(y)  \]

In the Appendix, using Mathematica, we prove that 
$\Phi_\beta(r)$ is decreasing on $\left[0,\frac{b}{
\beta}\right]$ and increasing on $\left[\frac{b}{\beta},\infty\right)$. Also at $y=b$, $r=\frac{b}{\beta}$ is the unique minimizing set for $\Phi_\beta$. That is, for $I_\beta(y(t)) = I_\beta(f^{-1}(t))$, we have
\[\inf_{\left|r-\frac{b}{\beta}\right|\ge\epsilon} I_\beta(r)=\min \left\{I_\beta\left(
y\left(\frac{b}{\beta}-\epsilon\right)\right),I_\beta\left(
y\left(\frac{b}{\beta}+\epsilon\right)\right)\right\}.\]

This implies $\Phi_\beta'\left(b\right)=0$. Furthermore, using Mathematica we can verify that $\Phi_\beta''\left(b\right)> 0$, which implies that there is a constant $C_{\beta}(N)$ such that
\[\inf_{\left|r-\frac{b}{\beta}\right|\ge\epsilon} I_\beta(r)\ge C_\beta (N) \epsilon^2,\]
which leads to
\[P_{n,\beta}\left[\left|\frac{|S_n|}{n}-\frac{b}{\beta}
\right|\ge\epsilon\right]\le e^{-C_\beta(N) n\epsilon^2}.\]

Now the approach similar to section 5 of \cite{KM}, where we apply above estimates along with $|S_n| \le n$ and $f(r) = \frac{r}{\beta}$ in equation (\ref{E:first-improvement}) leads to
\begin{equation}\begin{split}\label{E:third-improvement}
\E\left[W_n'-W_n\big|\sigma\right] &=-\frac{1-\beta
  f'(b)}{n}W_n+R,\end{split}\end{equation} where again
$\E|R|\le\frac{c_\beta(N)\log(n)}{n^{3/2}}.$ 
This completes the proof of part $(a)$.

\medskip

For part $(b)$, we will show the positivity of $\sigma^2$ and then we will use the asymptotics expanion of $|\sigma_i|$ and $\inprod{\sigma_i}{\sigma^{(i)}}$ in order to write the bound for second moment. Observe that by definition,
\begin{equation}\begin{split}\label{Delta^2-from-def}
\E\left[(W_n'-W_n)^2\big|\sigma\right]&=\frac{\beta^4}{n^4b^4}\sum_{i=1}^n\E
\left[\left.\left(2\sum_{j\neq i}\inprod{\sigma_i^*-\sigma_i}{\sigma_j}
\right)^2\right|\sigma,I=i\right]\\&=\frac{4\beta^4}{n^4b^4}\sum_{i=1}^n
\sum_{j,k\neq i}\E\left[\sigma_j^T(\sigma_i^*-\sigma_i)
(\sigma_i^*-\sigma_i)^T\sigma_k\big|\sigma,I=i\right].
\end{split}\end{equation}

Notice that here $\sigma_i^*$ is coming from the definition of the exchangabale pair $(W_n,~W_n')$. From the subcritical case calculations we have
\begin{equation*}\begin{split}
\E&\left[(\sigma_i^*-\sigma_i)
(\sigma_i^*-\sigma_i)^T\big|\sigma,I=i\right]
\\&\qquad=\left\{\left[\frac{I_{\frac{N}{2}}(\kappa)+\kappa I_{\frac{N}{2}+1}(\kappa)}{\kappa I_{ \frac{N}{2}-1}(\kappa)}\right]P_i+\left[\frac{I_{\frac{N}{2}}(\kappa)}{\kappa I_{\frac{N}{2}-1}(\kappa)}\right]P_i^\perp-\left[\frac{I_{\frac{N}{2}}(\kappa)}{ I_{\frac{N}{2}-1}(\kappa)}\right]
(r_i\sigma_i^T+\sigma_ir_i^T)+\sigma_i\sigma_i^T\right\}.
\end{split}\end{equation*}
where $\kappa=\frac{\beta|\sigma^{(i)}|}{n}$, $r_i=\frac{
\sigma^{(i)}}{|\sigma^{(i)}|}$, and $P_i$ is orthogonal projection onto$r_i$. Since $W_n$ is defined differently for supercritical case, we will later substitute modification of above expression into (\ref{Delta^2-from-def}).

In order to make sure that $\sigma^2 > 0$, we rewrite the first term of the last expression as
\[\frac{I_{\frac{N}{2}}(\kappa)+\kappa I_{\frac{N}{2}+1}(\kappa)}{\kappa I_{ \frac{N}{2}-1}(\kappa)} = 1-\frac{\kappa I_{ \frac{N}{2}-1}(\kappa)-I_{\frac{N}{2}}(\kappa)-\kappa I_{\frac{N}{2}+1}(\kappa)}{\kappa I_{ \frac{N}{2}-1}(\kappa)}. \]
Using series expansions of each term we have
\begin{equation}\label{E-Delta^2-approx}
\begin{split}\E\left[(\sigma_i^*-\sigma_i)
(\sigma_i^*-\sigma_i)^T\big|\sigma,I=i\right]&=
\left(1-\frac{N-1}{\beta}\right)P_i+\frac{1}{\beta}P_i^\perp-
\frac{b}{\beta}(r_i\sigma_i^T+\sigma_ir_i^T)+
\sigma_i\sigma_i^T+R_i'\\
&=\frac{1}{\beta}Id+\left(1-\frac{N}{\beta}\right)P_i-
\frac{b}{\beta}(r_i\sigma_i^T+\sigma_ir_i^T)+
\sigma_i\sigma_i^T+R_i',
\end{split}\end{equation}
where
\[R_i'=\left\{\left[\frac{I_{\frac{N}{2}}(\kappa)+\kappa I_{\frac{N}{2}+1}(\kappa)}{\kappa I_{ \frac{N}{2}-1}(\kappa)}-\frac{1}{\beta}\right]P_i+\left[\frac{I_{\frac{N}{2}}(\kappa)}{\kappa I_{\frac{N}{2}-1}(\kappa)}-\frac{1}{\beta}\right]P_i^\perp-\left[\frac{I_{\frac{N}{2}}(\kappa)}{ I_{\frac{N}{2}-1}(\kappa)}-\frac{b}{\beta}\right]
(r_i\sigma_i^T+\sigma_ir_i^T)\right\}.\]
Using the main term of 
\eqref{E-Delta^2-approx} into \eqref{Delta^2-from-def} yields

\begin{equation*}\begin{split}
&\E\left[(\sigma_i^*-\sigma_i)
(\sigma_i^*-\sigma_i)^T\big|\sigma,I=i\right]\\&\quad=\frac{4\beta^4}{n^4b^4}\sum_i\sum_{j,k\neq i}\left[
\frac{1}{\beta}\inprod{\sigma_j}{\sigma_k}+
\left(1-\frac{N}{\beta}\right)\sigma_j^TP_i\sigma_k-\frac{b}{\beta}(\sigma_j^Tr_i\sigma_i^T\sigma_k+\sigma_j^T\sigma_ir_i^T
\sigma_k)+\sigma_j^T\sigma_i\sigma_i^T\sigma_k\right].
\end{split}\end{equation*}
Following the calculations from section 5 of \cite{KM}, we get the following simplified form

\begin{equation*}\begin{split}
\E\left[(W_n'-W_n)^2\big|\sigma\right]&=\frac{4\beta^4}{n^4b^4}\sum_i
\left[\left(1-\frac{N-1}{\beta}\right)|\sigma^{(i)}|^2-\frac{2b}{\beta}
|\sigma^{(i)}|\inprod{\sigma_i}{\sigma^{(i)}}+\inprod{\sigma_i}{
\sigma^{(i)}}^2\right]\\&\qquad\qquad
+\frac{4\beta^4}{n^4b^4}\sum_i\sum_{j,k\neq i}\sigma_j^TR_i'\sigma_k.
\end{split}\end{equation*}

Now we have to find the deterministic constant which will be used to approximate the above final expression. Since
$|\sigma^{(i)}| \approx \frac{b(n-1)}{\beta}$ for each $i$, and 
$\inprod{\sigma_i}{\sigma^{(i)}}=\inprod{\sigma_i}{S_n}-1$,
this implies that $\inprod{\sigma_i}{\sigma^{(i)}}  \approx \frac{|S_n|^2}{n}-1\approx\frac{nb^2}{\beta^2}-1.$
We also have to rewrite the last expression in a deterministic way such that we can represent it in the form $2\lambda Var(\sigma)$ plus a mean zero term. It is important to note that this will help us to find the value of $Var(\sigma)$. 

Therefore we rewrite the above expression as:

\begin{equation}\begin{split}\label{2nd-diff-master}
\E&\left[(W_n'-W_n)^2\big|\sigma\right]\\&=\frac{4\beta^4}{n^3b^4}
\left[2\left(1-\frac{N-1}{\beta}\right)\frac{(n-1)^2b^2
}{\beta^2}-\frac{2n^2b^4}{\beta^4}
\right]+\frac{4\beta^4}{n^4b^4}\sum_i
\left(1-\frac{N-1}{\beta}\right)\left(|\sigma^{(i)}|^2-\frac{(n-1)^2b^2
}{\beta^2}\right)\\&\qquad\qquad+\frac{4\beta^4}{n^4b^4}\sum_i\left[
-\frac{2b}{\beta}\left(|\sigma^{(i)}|\inprod{\sigma_i}{
\sigma^{(i)}}-\frac{n^3b^3}{\beta^4}\right)+\inprod{\sigma_i}{
\sigma^{(i)}}^2-\left(1-\frac{N-1}{\beta}\right)\frac{(n-1)^2b^2
}{\beta^2}\right]\\&\qquad\qquad\qquad
+\frac{4\beta^4}{n^4b^4}\sum_i\sum_{j,k\neq i}\sigma_j^TR_i'\sigma_k,
\end{split}\end{equation}
and similar to \cite{KM} we define $\sigma$ such that
\[\frac{4\beta^4}{n^3b^4}
\left[2\left(1-\frac{N-1}{\beta}\right)\frac{(n-1)^2b^2
}{\beta^2}-\frac{2n^2b^4}{\beta^4}
\right]=2\lambda Var(\sigma),\]
where $\lambda=\frac{1-\beta  f'(b)}{n}$ defined as in part(a).
First we know that $f(r)=\frac{I_{\frac{N}{2}}(r)}{I_{\frac{N}{2}-1}(r)}\approx \frac{r}{\beta}$, consider the first term inside of brackets on R.H.S. of \ref{2nd-diff-master} (leading order in $n$),

\begin{equation*}\begin{split}
\frac{2n^2b^2}{\beta^2}\left[1-\frac{N-1}{\beta}-\frac{b^2}{\beta^2}\right]
&=\frac{2n^2b^2}{\beta^2}\left[1-\frac{(N-1)\left(\frac{I_{\frac{N}{2}}(b)}{I_{\frac{N}{2}-1}(b)}
\right)}{b}-\left(\frac{I_{\frac{N}{2}}(b)}{I_{\frac{N}{2}-1}(b)}\right)^2\right]\\&=\frac{2n^2b^2}{\beta^2}\left[1-\frac{N-1}{N}\frac{\left(I_{\frac{N}{2}-1}(b)+I_{\frac{N}{2}+1}(b)
\right)}{I_{\frac{N}{2}-1}(b)}-\left(\frac{I_{\frac{N}{2}}(b)}{I_{\frac{N}{2}-1}(b)}\right)^2\right]>0.
\end{split}\end{equation*}
The last conclusion can be verified using Mathematica and implies that $Var(\sigma)$ is positive and depends on $N$ and $\beta$. Here we use the following recurrence relation for modified Bessel function of first kind with $v=\frac{N}{2}$,
\[I_{v+1}(x)=\frac{2v}{x} I_v(x)-I_{v-1}(x)\]

This yield a strictly positive value of $Var(\sigma)$ which depends only on $\beta$ and $N$ and is independent of $n$. Now for applying Theorem from \cite{St} (p.35) we need to  estimate the expected absolute value of each of the terms above, which is straightforward and similar to \cite{KM}, calculation for corresponding section [5] except all the $1-\frac{2}{\beta}$ are replaced with $1-\frac{N-1}{\beta}$ which only changes the vaue of $c_\beta$.

Finally, part $(c)$ is trivial and similar to \cite{KM} section 5, with a different variance $Var(\sigma)$ coming from the corresponding hypersphere.
\qed
\end{proof}

\section{The Critical Temperature}\label{S:critical}

\begin{thm}\label{T:abstract_approx}
Consider an exchangable pair of positive random variables $(W,W')$. 
Assume that there exists a $\sigma$-field $\mathcal{F}\supseteq\sigma(W)$, such that 

\[\E\big[W'-W\big|\mathcal{F}\big]=Nk\big(1-cW^2\big)+R\]
and
\[\E\big[(W'-W)^2\big|\mathcal{F}\big]=kW+R'.\]
where $R$ and $R'$ are $\mathcal{F}$-measurable random variables and $k>0$  deterministic. Now consider a random variable $X$ with density function
\[p(t)=\begin{cases}\frac{1}{z} t^{\frac{N-2}{2}} e^{-\frac{t^2}{4 N^2 (N+2)} }&t\ge
0;\\0&t<0,\end{cases}\]
where $z$ is normalizing constant. Then there are constants $C_1,C_2,C_3$ such that for all $h\in C^2(\R)$,
\begin{equation*}\begin{split}
\big|\E h(W)-\E h(X)\big|&\le
\frac{C_1\|h\|_\infty}{k}\E|R|+\left(\frac{C_2(\|h\|_\infty+
\|h'\|_\infty)}{k}\right)\E|R'|\\&\qquad+
\left(\frac{C_3(\|h\|_\infty+\|h'\|_\infty+
\|h''\|_\infty)}{k}\right)\E|W'-W|^3.\end{split}\end{equation*}
\end{thm}
Construct an exchangable pair $(W_n,W_n')$, using Glauber dynamics, for the random variable 
\[W_n=\frac{c_N}{n^{3/2}}\sum_{i,j=1}^n\inprod{\sigma_i}{\sigma_j},\] 
which is defined in section 3. We obtain
\[W_n' = W_n-\frac{c_N}{n^{3/2}}\sum_{j=1}^n\inprod{\sigma_I}{\sigma_j}+\frac{c_N}{n^{3/2}}\sum_{j=1}^n \left\langle \sigma_I ',\sigma_j \right\rangle \] 
The following lemma gives the bounds needed to apply Theorem
\ref{T:abstract_approx} in this setting, and then Theorem \ref{T:limit_crit} follows immediately.
\begin{lemma}\label{L:crit_errors}
For a fixed $N$,  $(W_n,W_n')$ as constructed above, $k=\frac{c_N}{Nn^{3/2}}$ and
$c=\frac{N}{(N+2)c_N^2}$, we have
\begin{enumerate}
\item $\E\left[W_n'-W_n\big|\sigma\right]=Nk\left(1-cW_n^2\right)+R$
  and \, $\E|R|\le\frac{C(log(n))}{n^2}$;
\item $\E\left[(W_n'-W_n)^2\big|\sigma\right]=kW_n+R',$ and \,
  $\E|R'|\le\frac{C(log(n))}{n^2} $;
\item $\E|W_n'-W_n|^3\le\frac{C(log(n))}{n^{9/4}},$
\end{enumerate}
where C is a constant depending only on N, R and R' are defined below in the proof.
\end{lemma}
{\bf Proof of Lemma \ref{L:crit_errors}:}

For part $(a)$, similar to \cite{KM},
\begin{equation}\begin{split}\label{diff1_crit}
\E\left[W_n'-W_n\big|\sigma\right]&=-\frac{c_N}{n^{5/2}}\sum_{i=1}^n\left[
\sum_{k\neq i}\inprod{\sigma_i}{\sigma_k}-\E\left[\sum_{k\neq i}\inprod{\sigma_i
}{\sigma_k}\big|\{\sigma_j\}_{j\neq i}\right]\right]
\\&=-\frac{1}{n}W_n+\frac{c_N}{n^{3/2}}+\frac{c_N}{n^{5/2}}\sum_{i=1}^nf\left(\frac{N
|\sigma^{(i)}|}{n}\right)|\sigma^{(i)}|,
\end{split}\end{equation}
where $f(\kappa)=\frac{I_{\frac{N}{2}}(\kappa)}{I_{\frac{N}{2}-1}(\kappa)}$. 
Near zero, $f(\kappa) \approx \frac{\kappa}{N}-\frac{\kappa^3}{N^2(N+2)}$, so substituting this value in the above equation yields
\begin{equation}\begin{split}\label{f-expansion}\sum_{i=1}^nf\left(\frac{N
|\sigma^{(i)}|}{n}\right)|\sigma^{(i)}|=\sum_{i=1}^n\left[\frac{|\sigma^{(i)}|^2}{n}-
\frac{N|\sigma^{(i)}|^4}{(N+2)n^3}+O\left(\frac{|\sigma^{(i)}|^6}{n^5}\right)\right].
\end{split}\end{equation}
Note that
\[\frac{1}{n}\sum_{i=1}^n
|\sigma^{(i)}|^2=\frac{n^{3/2}W_n}{c_N}-\frac{2\sqrt{n}W_n}{c_N}+1.\]
Similarly the second term on the R.H.S. of equation (\ref{f-expansion}) is

\[-\frac{nNW_n^2}{(N+2)c_N^2}+\frac{4NW_n^2}{(N+2)c_N^2}-\frac{4N\sum_i\inprod{
\sigma_i}{S_n}^2}{(N+2)n^3}-\frac{2NW_n}{(N+2)c_N\sqrt{n}}+\frac{N W_n}{(N+2)n^{3/2}c_N}+\frac{N}{(N+2)n^2}.\]
Substituting these values into \eqref{f-expansion} and then simplified expression from there into \eqref{diff1_crit} yields

\[\E\left[W_n'-W_n\big|\sigma\right]= \frac{Nc_N}{Nn^{3/2}}-\frac{ N^2 W_n^2}{c_N N(N+2)n^{3/2}}+R ,\]
where 
\[R=-\frac{2W_n}{n^2}+\frac{c_N}{n^{5/2}}+\frac{4NW_n^2}{(N+2)c_N n^{5/2}}-\frac{4 c_N N\sum_i\inprod{
\sigma_i}{S_n}^2}{(N+2)n^{11/2}}-\frac{2NW_n}{(N+2)n^3}+\frac{N W_n}{(N+2)n^{4}}+\frac{c_N N}{(N+2)n^{9/2}},\]
Therefore we have,
\[\E\left[W_n'-W_n\big|\sigma\right]= Nk\left(1-cW_n^2\right)+R,\] 
where $ k =\frac{c_N}{Nn^{3/2}}$ , $c = \frac{N}{(N+2)c_N^2}$ and $\E|R|\le
\frac{C(log(n))}{n^2} .$
\medskip
\\For part $(b)$, from the
definition as before,
\begin{equation}\begin{split}\label{Delta^2-from-def-crit}
\E\left[(W_n'-W_n)^2\big|\sigma\right]&=\frac{c_N^2}{n^4}\sum_{i=1}^n
\sum_{j,k\neq i}\E\left[\sigma_j^T(\sigma_i^*-\sigma_i)
(\sigma_i^*-\sigma_i)^T\sigma_k\big|\sigma,I=i\right].
\end{split}\end{equation}
Using a previous computation, the terms $\E\left[(\sigma_i^*-\sigma_i)
(\sigma_i^*-\sigma_i)^T\big|\sigma,I=i\right] $ are
\begin{equation*}\begin{split}
\left(\frac{a^2}{n^2}\right)\sum_{i=1}^n\left\{\left[\frac{I_{\frac{N}{2}}(\kappa)+\kappa I_{\frac{N}{2}+1}(\kappa)}{\kappa I_{ \frac{N}{2}-1}(\kappa)}\right]P_i+\left[\frac{I_{\frac{N}{2}}(\kappa)}{\kappa I_{\frac{N}{2}-1}(\kappa)}\right]P_i^\perp -\left[\frac{I_{\frac{N}{2}}(\kappa)}{ I_{\frac{N}{2}-1}(\kappa)}\right]
(r_i\sigma_i^T+\sigma_ir_i^T)+\sigma_i\sigma_i^T\right\}.
\end{split}\end{equation*}
Again near zero, $f(\kappa)\approx \frac{\kappa}{N}-\frac{\kappa^3}{N^2(N+2)}$, so
\begin{equation}\label{E-Delta^2-approx-crit}
\begin{split}\E\left[(\sigma_i^*-\sigma_i)
(\sigma_i^*-\sigma_i)^T\big|\sigma,I=i\right]
&=\frac{1}{N}Id-
\frac{\kappa_i}{N}(r_i\sigma_i^T+\sigma_ir_i^T)+
\sigma_i\sigma_i^T+R_i',
\end{split}\end{equation}
where 
\[R_i'=\left(\frac{f(\kappa_i)}{\kappa_i}-\frac{1}{N}\right)Id-\left(f(\kappa_i)-\frac{\kappa_i}{N}\right)
\left(r_i\sigma_i^T+\sigma_ir_i^T\right).\]
Ignoring the $R_i'$ for the moment and putting the main term of 
\eqref{E-Delta^2-approx-crit} into \eqref{Delta^2-from-def-crit} yields
\begin{equation*}\begin{split}
\frac{c_N^2}{n^4}\sum_i\sum_{j,k\neq i}&\left[
\frac{1}{N}\inprod{\sigma_j}{\sigma_k}-
\frac{\kappa_i}{N}(\sigma_j^Tr_i\sigma_i^T\sigma_k+\sigma_j^T\sigma_ir_i^T
\sigma_k)+\sigma_j^T\sigma_i\sigma_i^T\sigma_k\right].
\end{split}\end{equation*}
\begin{equation*}\begin{split}
\E\left[(W_n'-W_n)^2\big|\sigma\right]&=\frac{c_N^2}{n^4}\sum_i
\left[\frac{1}{N}|\sigma^{(i)}|^2-\frac{2}{n}
|\sigma^{(i)}|^2\inprod{\sigma_i}{\sigma^{(i)}}+\inprod{\sigma_i}{
\sigma^{(i)}}^2\right]\\&\qquad\qquad
+\frac{c_N^2}{n^4}\sum_i\sum_{j,k\neq i}\sigma_j^TR_i'\sigma_k\\
&=\frac{c_N^2}{n^4}\sum_i
\left[\frac{2}{N}|\sigma^{(i)}|^2-\frac{2}{n}
|\sigma^{(i)}|^2\inprod{\sigma_i}{\sigma^{(i)}}\right]
+\frac{c_N^2}{n^4}\sum_i\sum_{j,k\neq i}\sigma_j^TR_i'\sigma_k\\&=
\frac{2c_N^2}{Nn^3}\left(\frac{n^{3/2}W_n}{c_N}-\frac{2\sqrt{n}W_n}{c_N}+1\right)-\frac{2c_N^2}{n^5}\sum_i
|\sigma^{(i)}|^2\inprod{\sigma_i}{\sigma^{(i)}}\\+ \frac{c_N^2}{n^4}\sum_i\sum_{j,k\neq i}\sigma_j^TR_i'\sigma_k,
\end{split}\end{equation*}
where the computation for
$\E\left[\inprod{\sigma_i}{\sigma^{(i)}}^2\right]$ from the
supercritical case has been used.
Recall that the main term should be $\frac {1}{N}kW_n=\frac{c_NW_n}{Nn^{3/2}}$ and
indeed it is.  It is a routine collection of arguments very similar to
those in the previous sections to show that the remaining terms are
bounded in expectation by $\frac{C\log(n)}{n^2}$.

Finally, part $(c)$ is straightforward as usual and same as \cite{KM} section 06(c) with different $\sigma$ belonging to corresponding unit hypersphere.

\qed

\section{Appendix}
{\bf\subsection{ Calculus of $\Phi_\beta$}}

We will start this section by revisiting that the free energy can be obtained by minimizing the functional

\[\Phi_\beta(r)= r \frac {I_\frac{N}{2}(r)}{I_{\frac{N}{2}-1}(r)}+\log\left[\frac{A_{N}}{A_{N-1}} \frac {r^{\frac{N}{2}-1}}{B_N \pi I_{\frac{N}{2}-1}(r) }\right]-\frac{\beta}{2}\left(\frac {I_\frac{N}{2}(r)}{I_{\frac{N}{2}-1}(r)}\right)^2 .\]
where
\[B_N= \begin{cases} \prod_{k=0}^{\frac{N}{2}-1} |2k-1| , \quad \quad \quad \quad \quad \quad \quad \quad \text{ if N even}, \\ \\ 2^{\frac{N-3}{2}} (1)_{\frac{N-3}{2}},  \quad \quad \quad \quad \quad \quad \quad \quad \quad\;\, \text{ if N odd}, \end{cases}\]
and  $r = g^{-1}(\beta)$, with \[g(r) = g_N(r) := r \frac{ I_{\frac{N}{2}-1}(r) }{ I_\frac{N}{2}(r)}\] 

\begin{lemma}\label{cal-phi}
Consider the functional defined above:
\begin{enumerate}
\item For $\beta\le N$, the
$\inf_{b\ge 0}\left\{\Phi_\beta(b)\right\}=0$
achieved only at $b=0$.

\item For $\beta>N$, there is a unique value of $r\in(0,\infty)$
which minimizes $\Phi_\beta$ over $[0,\infty)$.

\item Let $b$ denote the unique positive solution of $ b-\beta f(b)=0 $ with \[f(b) = \frac{ I_{\frac{N}{2}}(b) }{ I_{\frac{N}{2}-1}(b)}\]
Then $\beta f'(b) <1.$
In particular, $\Phi_\beta'(b)=0$ and $\Phi_\beta''(b)> 0$.

\end{enumerate}

\end{lemma}
{\bf Proof of Lemma \ref{cal-phi}:}
\begin{enumerate}
 \item  Using Mathematica we can check that $ \Phi_\beta(b)$ is increasing on $ (0,\infty)$ for $\beta \leq N$. Also 
\[ \lim_{b \to 0} \frac {I_\frac{N}{2}(b)}{I_{\frac{N}{2}-1}(b)}=0\] and
\[  \lim_{b \to 0} \log\left[\frac{A_{N}}{A_{N-1}} \frac {b^{\frac{N}{2}-1}}{B_N \pi I_{\frac{N}{2}-1}(b) }\right] =0.\]
Therefore,

\[ \lim_{b \to 0} \Phi_\beta(b)=0.\]

\item Since we have given the results for generalized dimensions, here we are giving graphical proof  for $N=2$ case (Figure \ref{fig:2}). Using some mathemcatical software such as Mathematica we can verify similar for higher dimensions as well. 

\begin{figure}
\centering
\includegraphics[width=4in]{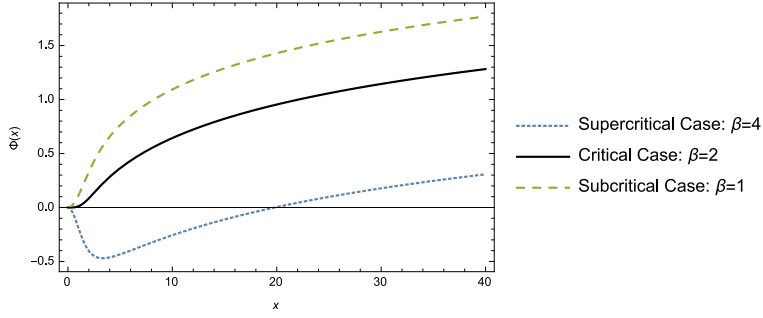}
\caption{Graphical representation of functional $\Phi_\beta(x)$ for the mean-field XY Model}
\label{fig:2}
\end{figure}

\item For this part, again we can check using Mathematica that this statement is true for all dimensions.

\end{enumerate}
\qed
\subsection{{\bf Proof of Theorem \ref{T:abstract_approx}}}\label{Stein_app}
In this section we will give the proof for Theorem \ref{T:abstract_approx}, with two lemmas which are needed. It is to be noted that for applicatons of Stein's method we need to identify the characteristic operator and density of the distributions.

\begin{lemma}\label{L:char}

A random variable $Y > 0 $ has density
\[p(t)=\begin{cases} \frac{1}{Z}t^{\frac{N-2}{2}} e^{-\widetilde{k} t^2} &t\ge
   0\\0&t<0\end{cases}\]
if and only if 
\begin{equation}\label{char-eq}\E\left[Yf'(Y)+\left( \frac{N}{2}-2 \widetilde{k} Y^2\right) f(Y)\right]=0
\end{equation}
for $\widetilde{k} = \frac{1}{N^2(4N+8)}$ and all $f\in C^1((0,\infty))$ such that $\int_0^\infty f(t)p(t)dt<\infty$.
The corresponding related distribution has a characterizing operator $T_p$ which is invertible on space $\{h: \E h(x) = 0\}$ and defined by:
\[[T_p f](x)=xf'(x)+\left(\frac{N}{2}-2 \widetilde{k} x^2\right) f(x).\]
\end{lemma}
{\bf Proof of Lemma \ref{L:char}:}\\

Consider a positive random variable $Y$ which has $p(t)$ as its density function, then using integration by parts it is straightforward to show that $Y$ satisfies \eqref{char-eq}.

Conversely,  consider a random variable $X$ having $p(t)$ as its density function. Then given $h:(0,\infty)\to\R$, we construct $f=f_h$ so that
\[tf'(t)+\left(\frac{N}{2}-2 \widetilde{k} t^2\right) f(t)=h(t)-\E h(X).\]  We claim that the solution $f$ is given by 
\begin{equation*}\begin{split}f(t)&=\frac{1}{tp(t)}\int_0^t\big[h(s)-\E
    h(X)\big]p(s)ds \\&=-\frac{1}{tp(t)}\int_t^\infty \big[h(s)-\E h(X)\big]p(s)ds.\end{split}\end{equation*}
To see this, we differentiate the expression similar to \cite{KM} and deduce that
\[h(t)-\E
h(X)=f(t)+tf'(t)+\frac{tf(t)p'(t)}{p(t)}=\left(\frac{N}{2}-2 \widetilde{k} t^2\right) f(t)+tf'(t).\]

Therefore for bounded $f$ and $f'$ and $Y$ satisfying \eqref{char-eq}, 
then for given $h$, $f=f_h$ solves the Stein equation, 
\[\E h(Y)-\E h(X)=\E\left[Yf'(Y)+\left(\frac{N}{2}-2 \widetilde{k} Y^2\right) f(Y)\right]=0,\]
thus $Y\overset{d}{=}X.$
\qed

\begin{lemma}\label{boundedness}
The characteristic operator defined above has the following boundedness results: Let $h:\R\to\R$ be given.  Suppose that
\[f(t)=f_h(t):=\frac{1}{tp(t)}\int_0^t\big[h(s)-\E
    h(X)\big]p(s)ds,\]
with $p$ as defined in the previous lemma and $\widetilde{k} = \frac{1}{N^2(4N+8)}$.  Then $[T_pf_h](x)=h(x)-\E h(X)$ and 
\begin{enumerate}
\item $\|f_h\|_\infty\le \left(4 N^5 e\right)^{\frac{N}{4}}\|h\|_\infty .$
\item
  $\|f_h'\|_\infty\le \le \left( 2 \left(4 N^5 e\right)^{\frac{N}{4}}+  N \right) \widetilde{k} t_0 \|h\|_\infty + \frac {N}{4}\|h'\|_\infty.$
\item
  $\|f_h''\|_\infty\le K_1\|h\|_\infty+K_2\|h'\|_\infty+
K_3\|h''\|_\infty,$ where $K_1,K_2,K_3$ are constants depending on dimension $N$ and $c_N$.
\end{enumerate}
\end{lemma}

{\bf Proof of Lemma \ref{boundedness}:}
\begin{enumerate}
\item The first formula for $f_h$, gives the following bound:
\[f(t)\le\frac{2\|h\|_\infty}{tp(t)}\left(\int_0^tp(s)ds\right).\]

For $t \le t_0 = \sqrt{\frac{N}{4 \widetilde{k}}}$, using Mathematica we can check that  
\\ \[\frac {\int_0^tp(s)ds}{t p(t)} \le \frac{\left(4 N^5 e\right)^{\frac{N}{4}}}{2} \]

Therefore we have
\[f(t)\le  \left(4 N^5 e\right)^{\frac{N}{4}}\|h\|_\infty .\]
Also using the fact that $X$ has density function $p(t)$ and from the definition of $f$, similar approach as in \cite{KM} can be used to find bound on $|f(t)|$. For any fixed $N$ and $t\geq t_0$, we have \[ |f(t)| \leq \|h\|_\infty.\]

\item We know that $f$ solves the Stein equation,
\[tf'(t) = \left(2 \widetilde{k} t^2-\frac{N}{2}\right) f(t)+h(t)-\E h(X).\]
For $t\le t_0$, we have 
\begin{equation*}\begin{split}
h(t)-\E h(X)-\frac{N}{2} f(t)&=h(t)-\E h(X)-\frac{N}{2tp(t)}\int_0^t[h(s)-\E
h(X)]p(s)ds\\&=\frac{N}{2tp(t)}\int_0^t\Big([h(t)-\E h(X)]\left( \frac{s^{\frac{N-2}{2}}}{t^{\frac{N-2}{2}}}\right) p(t)-[h(s)-\E
h(X)]p(s)\Big)ds.
\end{split}\end{equation*}
Now, notice that  $\frac{p(s)}{p(t)}\le 1 $ so we have

\begin{equation*}\begin{split}
&\left|\frac{1}{tp(t)}\int_0^t[h(t)-\E h(X)]\left(\left( \frac{s}{t}\right)^{\frac{N-2}{2}}p(t)-p(s)\right)ds\right|\le
\frac{2\|h\|_\infty}{tp(t)}\int_0^t\left|1-\left( \frac{s}{t}\right)^{\frac{N-2}{2}}\frac{p(t)}{p(s)}\right|
p(s)ds \\& \qquad\qquad\qquad \qquad\qquad \qquad\qquad\qquad \qquad=\frac{2\|h\|_\infty}{tp(t)}\int_0^t\left|1-e^{\widetilde{k}(s^2-t^2)}\right|
p(s)ds \\&\qquad\qquad\qquad \qquad\qquad \qquad\qquad\qquad \qquad
\le 2 \widetilde{k}\|h\|_\infty t^2.
\end{split}\end{equation*}
Also,
\begin{equation*}\begin{split}
\left|\frac{1}{tp(t)}\int_0^t\Big([h(t)-\E h(X)]-[h(s)-\E
  h(X)]\Big)p(s)ds\right|
\le\frac{\|h'\|_\infty}{tp(t)}\int_0^t(t-s)p(s)ds\le
\frac{\|h'\|_\infty t}{2}.
\end{split}\end{equation*}
This implies for $t\le t_0$, we have
\[\frac{1}{t}\Big|h(t)-\E h(X)-\frac{N}{2}f(t)\Big|\le  N \widetilde{k} t_0 \|h\|_\infty + \frac {N}{4}\|h'\|_\infty,\]
and since  
\[f'(t) = 2 \widetilde{k} t f(t)+\frac{1}{t}\left(h(t)-\E h(X)-f(t) \right),\]
as a result, we have
\begin{equation*}\begin{split}
\big|f'(t)\big|&\le
2 \widetilde{k} t_0 \|f\|_\infty+  N \widetilde{k} t_0 \|h\|_\infty + \frac {N}{4}\|h'\|_\infty\\& \le 2 \widetilde{k} t_0 \left(4 N^5 e\right)^{\frac{N}{4}}\|h\|_\infty +  N \widetilde{k} t_0 \|h\|_\infty + \frac {N}{4}\|h'\|_\infty\\& \le \left( 2 \left(4 N^5 e\right)^{\frac{N}{4}}+  N \right) \widetilde{k} t_0 \|h\|_\infty + \frac {N}{4}\|h'\|_\infty.\end{split}\end{equation*}

For $t \geq t_0 = \sqrt{\frac{N}{4 \widetilde{k}}}$, from Stein equation we get:
\begin{equation*}\begin{split}
|f'(t)| & \leq 2 \widetilde{k} t |f(t)| + \frac{ \frac{N}{2}\|f\|_\infty + 2 \|h\|_\infty }{t} 
         \\ & \leq  \frac{4 \widetilde{k}  \|h\|_\infty P\left[X \ge t\right]}{p(t)} + \frac{ \frac{N}{2} \|f\|_\infty + 2 \|h\|_\infty }{t} .
\end{split}\end{equation*}
Using the estimate \[P\left[X \ge t\right] \le \frac{\Gamma{\left( \frac{N}{4},\frac{t^2}{4N^2(N+2)}\right) }}{\Gamma{\left( \frac{N}{4}\right) }}\] along with some simplifications completes the proof.

\item 
Consider again the Stein equation 
\[f'(t) = \left(2 \widetilde{k} t-\frac{N}{2t}\right) f(t)+\frac{1}{t}\left(h(t)-\E h(X)\right).\]

differentiating both sides with respect to $t$ and substituting value of $f'(t)$ from above, we obtain
\begin{align*}
 f''(t) & = 2 \widetilde{k} \left( f(t)+t f'(t)\right)+\frac{N}{2 t^2} f(t)-\frac{N}{2t} f'(t)-\frac{1}{t^2}\left(h(t)-\E h(X)\right)+\frac{h'(t)}{t} \\& = 2 \widetilde{k} \left( f(t)+t f'(t)\right)-\frac{N}{2t} f'(t)+\frac{1}{t}\left(h'(t)+\frac{\frac{N}{2} f(t)-\left[h(t)-\E h(X)\right]}{t}\right) .
\end{align*}

Using a similar approach to \cite{KM}, the last term from above simplifies to

\begin{align*}
& h'(t)-\frac{\left[h(t)-\E h(X)\right]-\frac{N}{2} f(t)}{t}\\& = h'(t)-\frac{N}{2t^2p(t)}\int_0^t\Big([h(t)-\E h(X)]\left(\frac{s}{t}\right)^{\frac{N-2}{2}} p(t)-[h(s)-\E h(X)]p(s)\Big)ds \\& = -\frac{N}{2t^2p(t)}\int_0^t\Big([h(t)-\E h(X)]\left(\frac{s}{t}\right)^{\frac{N-2}{2}} p(t)-[h(s)-\E h(X)]p(s)-\left(\frac{s}{t}\right)^{\frac{N-2}{2}} (t-s) h'(t) p(t) \Big)ds \\&= -\frac{N}{2t^2p(t)}\int_0^t\Big([h(t)-\E h(X)]-[h(s)-\E h(X)]-(t-s) h'(t) \Big) \frac{e^{-\widetilde{k} t^2} s^{\frac{N-2}{2}}}{z}ds  .
\end{align*}

Define \[H(t) = \left[h(t)-\E h(X)\right]p(t)  ,\]
then \[(t-s) h'(t) = (t-s) H'(t) + 2 \widetilde{k} t [h(t)-\E h(X)]\]
Then the above simplifies to
\begin{align*}
& h'(t)- \frac{\left[h(t)-\E h(X)-\frac{N}{2} f(t)\right] }{t} \\& = \frac{2}{t^{2} p(t)} \int_{0}^{t} \left(H(s)-H(t)-(s-t) H'(t) + 2 \widetilde{k} t \left[h(t)-\E h(X)\right] \right)  \frac{e^{-\widetilde{k} t^2} s^{\frac{N-2}{2}}}{z} ds.
\end{align*}

The rest of this part is computing the bounds for all terms similar to \cite{KM}, which we leave for the reader to check.
\qed

\end{enumerate}
{\bf Proof of Theorem \ref{T:abstract_approx}:}

Given $h$, let $f$ be the solution to the Stein equation described
above.  Then by exchangeability and the conditions on $(W,W')$,
\begin{equation*}\begin{split}
0&=\E[(W'-W)(f(W')+f(W))]\\&=\E[(W'-W)(f(W')-f(W))+2(W'-W)f(W)]\\
&=\E[(W'-W)^2f'(W)+E''+2Nk(1-cW^2)f(W)+2Rf(W)]\\
&=\E[kWf'(W)+R'f'(W)+E''+2Nk(1-cW^2)f(W)+2Rf(W)].
\end{split}\end{equation*}
where 
\[E'' = \sum_{n=2}^{\infty} \frac{f^{n}(W)}{n!} (W'-W)^{n+1} \]
Then
\[\E[Wf'(W)+2N(1-cW^2)f(W)]=-\frac{1}{k}\E[R'f'(W)+2Rf(W)+E''],\]
and
\[|E''|\le \frac{\|f''\|_\infty}{2}|(W'-W)|^3.\]
The result is thus immediate from Lemma \ref{boundedness}.

\qed
\subsection{{\bf Density Function Calculation}}
Consider the functional $G_\beta(r) $ defined as:
\[G_\beta(r)= r \frac {I_\frac{N}{2}(r)}{I_{\frac{N}{2}-1}(r)}+\log\left[\frac{A_{N}}{A_{N-1}} \frac {r^{\frac{N}{2}-1}}{B_N\pi I_{\frac{N}{2}-1}(r) }\right]-\frac{\beta}{2}\left(\frac {I_\frac{N}{2}(r)}{I_{\frac{N}{2}-1}(r)}\right)^2 ,\]
where 
\[B_N= \begin{cases} \prod_{k=0}^{\frac{N}{2}-1} |2k-1| , \quad \quad \quad \quad \quad \quad \quad \quad \text{ if N even}, \\ \\ 2^{\frac{N-3}{2}} (1)_{\frac{N-3}{2}},  \quad \quad \quad \quad \quad \quad \quad \quad \quad\;\, \text{ if N odd}. \end{cases}\]
We can check that $r = 0$ is point of inflection for $G(r)$ at the critical value $\beta = N$. We can write the Taylor expansion for the critical case as follows:
\[G(r)= G(m)+ \lambda(m)\frac{(r-m)^{2k+1}}{(2k+1)!}+O(r)^{5},\]
where $ m = 0 $ , $k = \frac{3}{2}$, $G(0) = 0$ and $\lambda(0) = \frac{3!}{N^{2} (2+N)}$.  Finally from Theorem 5 of \cite{RECN}, our density function for $r \ge 0$ is given by
\begin{equation}
	\begin{split}
p(r)=\frac{1}{\widetilde{z}} r^{N-1} e^{-\widetilde{k} r^4},
	\end{split}
	\end{equation}
with $\widetilde{k} = \frac{1}{4 N^2 (N+2)}$. Using substitution $t = r^2$ we obtain:
\begin{equation}
	\begin{split}
p(t)=\frac{1}{z} t^{\frac{N-2}{2}} e^{-\widetilde{k} t^2},
	\end{split}
	\end{equation}
Therefore, the density function at the critical temperature for the $O(N)$-model is given by
 \begin{equation} \label {DF}
 \begin{split}
p(t)=\begin{cases}\frac{1}{z} t^{\frac{N-2}{2}} e^{-\widetilde{k} t^2}&t\ge
0;\\0&t<0,\end{cases}
 \end{split}
 \end{equation}
 
 The reader can also verify the above density function using an approach similar to \cite{SDHR}.
{\bf acknowledgements}
The authors wish to thank Richard Ellis, Charles Newman, Enzo Marinari, and Leslie A. Ross for helpful discussions.




%
%

\end{document}